\documentclass[11pt,a4paper]{article}
\usepackage[margin=35mm]{geometry}

\usepackage{amssymb,amsmath,amsfonts,amsthm, amscd, mathrsfs, helvet, mathtools}
\usepackage{bm, stmaryrd}
\usepackage{graphicx, verbatim}
\usepackage[mathcal]{euscript}
\usepackage{color}
\usepackage[all]{xy}
\usepackage{relsize}
\usepackage{bbold}
\renewenvironment{comment}{}{}
\usepackage{comment}

\usepackage{tikz}
\usetikzlibrary{cd}
\usetikzlibrary{matrix,calc}
\usetikzlibrary{decorations.markings,shapes.geometric,shapes.misc}
\usetikzlibrary{decorations.pathreplacing,positioning}
\usetikzlibrary{arrows}

\tikzset{cross/.style={cross out, draw=black, minimum size=2*(#1-\pgflinewidth), inner sep=0pt, outer sep=0pt}, cross/.default={1pt}}

\tikzset{cross/.style={cross out, draw=black, minimum size=2*(#1-\pgflinewidth), inner sep=0pt, outer sep=0pt}, cross/.default={1pt}}

\allowdisplaybreaks[1]

\usepackage{pict2e}

\usepackage[pdfencoding=auto,psdextra,hidelinks]{hyperref}
\hypersetup{
     colorlinks=false,         
     linkcolor=darkred,
     citecolor=blue,
}

\makeatletter
\DeclareFontFamily{OMX}{MnSymbolE}{}
\DeclareSymbolFont{MnLargeSymbols}{OMX}{MnSymbolE}{m}{n}
\SetSymbolFont{MnLargeSymbols}{bold}{OMX}{MnSymbolE}{b}{n}
\DeclareFontShape{OMX}{MnSymbolE}{m}{n}{
    <-6>  MnSymbolE5
   <6-7>  MnSymbolE6
   <7-8>  MnSymbolE7
   <8-9>  MnSymbolE8
   <9-10> MnSymbolE9
  <10-12> MnSymbolE10
  <12->   MnSymbolE12
}{}
\DeclareFontShape{OMX}{MnSymbolE}{b}{n}{
    <-6>  MnSymbolE-Bold5
   <6-7>  MnSymbolE-Bold6
   <7-8>  MnSymbolE-Bold7
   <8-9>  MnSymbolE-Bold8
   <9-10> MnSymbolE-Bold9
  <10-12> MnSymbolE-Bold10
  <12->   MnSymbolE-Bold12
}{}

\DeclareMathAlphabet{\mathscrbf}{OMS}{mdugm}{b}{n}

\let\llangle\@undefined
\let\rrangle\@undefined
\DeclareMathDelimiter{\llangle}{\mathopen}%
                     {MnLargeSymbols}{'164}{MnLargeSymbols}{'164}
\DeclareMathDelimiter{\rrangle}{\mathclose}%
                     {MnLargeSymbols}{'171}{MnLargeSymbols}{'171}
\makeatother



\definecolor{myPurple}{rgb}{0.5,0.1,0.6}
\definecolor{myOrange}{rgb}{1.0,0.5,0.0}
\definecolor{myRed}{rgb}{1.0,0.0,0.0}
\definecolor{myGreen}{rgb}{0.0,0.5,0.0}
\definecolor{LatexBlue}{rgb}{0.211765,0.227451,0.666667}
\definecolor{myBlue}{rgb}{0.0,0.0,1.0}
\definecolor{myBlack}{rgb}{0.0,0.0,0.0}
\definecolor{myGray}{rgb}{0.3,0.3,0.3}

\theoremstyle{plain}
\newtheorem{theorem}{Theorem}[section]
\newtheorem*{theorem*}{Theorem}
\newtheorem{proposition}[theorem]{Proposition}
\newtheorem*{proposition*}{Proposition}

\newtheorem{corollary}[theorem]{Corollary}

\theoremstyle{definition}

\newtheorem{remarkx}[theorem]{Remark}

\newenvironment{remark}
  {\pushQED{\qed}\remarkx}
  {\popQED\endremarkx}

\DeclareMathOperator{\res}{res}

\DeclareMathOperator{\Tr}{Tr}

\newcommand{\lda}{\lambda}

\newcommand{\g}{\mathfrak{g}}
\newcommand{\Lg}{\widetilde{\mathfrak{g}}}
\newcommand{\LG}{\widetilde{G}}

\newcommand{\Lag}{\mathscr{L}}

\newcommand{\s}{\sigma}

\def\sl{\mathfrak{sl}}

\def\s{\mathfrak{s}}

\newcommand{\lau}[1]{(\kern-.2em( #1 )\kern-.2em)}

\newcommand{\ie}{{\it i.e.}\ }

\def\be{\begin{equation}}
\def\ee{\end{equation}}
\def\bea{\begin{eqnarray}}
\def\eea{\end{eqnarray}}

\def\RR{\mathbb{R}}

\def\1{\bm{1}}

\numberwithin{equation}{section}

\begin{document}
	
\begin{center}
{\LARGE\bf Lagrangian multiforms on coadjoint orbits for finite-dimensional integrable systems}\\
\end{center}

\vspace{1cm}
\begin{center}
{\large Vincent Caudrelier$^{1}$, Marta Dell'Atti$^{2}$, Anup Anand Singh$^{1}$}
\end{center}

\begin{center}
{\it $^{1}$ School of Mathematics, University of Leeds, UK\\
$^{2}$ School of Mathematics and Physics, University of Portsmouth, UK}
\end{center}

\vspace{3cm}

\begin{abstract}
Lagrangian multiforms provide a variational framework to describe integrable hierarchies. The case of Lagrangian $1$-forms covers finite-dimensional integrable systems. We use the theory of Lie dialgebras introduced by Semenov-Tian-Shansky to construct a Lagrangian $1$-form. Given a Lie dialgebra associated with a Lie algebra $\g$ and a collection $H_k$, $k=1,\dots,N$, of invariant functions on $\g^*$, we give a formula for a Lagrangian multiform describing the commuting flows for $H_k$ on a coadjoint orbit in $\g^*$. 
We show that the Euler-Lagrange equations for our multiform produce the set of compatible equations in Lax form associated with the underlying $r$-matrix of the Lie dialgebra. We establish a structural result which relates the closure relation for our multiform to the Poisson involutivity of the Hamiltonians $H_k$ and the so-called ``double zero'' on the Euler-Lagrange equations.
 The construction is extended to a general coadjoint orbit by using reduction from the free motion of the cotangent bundle of a Lie group.
 We illustrate the dialgebra construction of a Lagrangian multiform with the open Toda chain and the rational Gaudin model. The open Toda chain is built using two different Lie dialgebra structures on $\sl(N+1)$. The first one possesses a non-skew-symmetric $r$-matrix and falls within the Adler--Kostant--Symes scheme. The second one possesses a skew-symmetric $r$-matrix. In both cases, the connection with the well-known descriptions of the chain in Flaschka and canonical coordinates is provided. 
\end{abstract}

\newpage

\tableofcontents

\section{Introduction}

 The concept of Lagrangian multiforms was introduced in \cite{LN} with the objective of providing a variational criterion of integrability. The pioneering insight was inspired by the well-established criterion for integrability known as multidimensional consistency \cite{N1,ABS} which is the discrete analogue of the property of commuting (Hamiltonian) flows for a dynamical system sitting in an integrable hierarchy. It was proposed to introduce a generalised action and a variational principle, involving a new object (a Lagrangian multiform), to capture purely variationally multidimensional consistency. This idea grew quickly, first in the discrete realm, see \cite{HJN} and references therein. Over the last decade or so, the universality of this idea and its connections with more traditional features of integrability (Lax pair, Hamiltonian structures) has been illustrated in many other incarnations of integrable systems: finite-dimensional systems \cite{YKLN} followed by \cite{Su1,PS}, continuous infinite-dimensional systems -- field theories in $1+1$ dimensions \cite{SV,SNC1,SNC2,CS1,CS2,PV,CSV} and in $2+1$ dimensions \cite{SNC3,N2} -- and semi-discrete systems \cite{SV2}. The relations between discrete and continuous multiforms were explored in \cite{Ver}. The concept was even extended recently to non-commuting flows in \cite{CNSV}. In general, a Lagrangian multiform is a $d$-form which is integrated over a hypersurface of dimension $d$ in a so-called multi-time space of dimension greater than $d$ to yield an action functional depending not only on the field configurations but also on the hypersurface. This last point is the main departure from a traditional action and principle of least action. One postulates a principle of least action which must be valid for {\it any} hypersurface embedded in the multi-time space. This is the postulate which captures the idea of the commutativity of the flows and which was adopted as a definition of pluri-Lagrangians, see \cite{PS,PV} and references therein. In Lagrangian multiform theory, there is an additional postulate, the {\it closure relation} which is the direct counterpart of the Poisson involutivity of Hamiltonians, the Liouville criterion for integrability. 

The generalised variational principle produces equations that come in two flavours: 1) Euler-Lagrange equations associated with each of the coefficients of the Lagrangian multiform which form a collection of Lagrangian densities; 2) Corner or structure equations on the Lagrangian coefficients themselves which select possible models and ensure the compatibility of the various equations of motion imposed on a common set of fields. Classifying all possible Lagrangian multiforms along these lines would amount to classifying all integrable hierarchies. In practice, it is a nontrivial task to obtain all the Lagrangian coefficients of a multiform which produce compatible equations of motion. Beyond brute-force calculations to solve the corner equations~\cite{SNC1}, several works have used the idea of variational symmetries to achieve this goal \cite{PS,PV,SNC2}. This produces an algorithm to construct the Lagrangian coefficients one after the other from a given initial Lagrangian. Although perfectly fine in theory, this can become quickly unmanageable in practice and usually formulas for only a few Lagrangian coefficients are obtained. It also has the disadvantage of singling out some independent variables in the hierarchy which then appear as the so-called ``alien derivatives'' in the higher Lagrangian coefficients.  

More recently, another approach was introduced which takes a more global view on a hierarchy and provides an efficient way of describing all the Lagrangian coefficients in one formula~\cite{CS2,CSV}, see also \cite{N2}. A key insight in \cite{CS2,CSV} was the incorporation in the Lagrangian multiform of key ingredients known in the Hamiltonian framework for integrable hierarchies, in particular the classical $r$-matrix, as well as the ``compounding'' of hierarchies following \cite{N3}. This paper draws and expands upon this insight and is concerned with Lagrangian $1$-forms which allow one to treat integrable hierarchies of {\it finite}-dimensional systems. Specifically, we show how the theory of {\it Lie dialgebras} \cite{STS} can be used to construct systematically a Lagrangian multiform for {\it any} finite-dimensional system which falls within the Lie dialgebra framework. The latter incorporates and generalises the perhaps more well-known Adler-Kostant-Symes scheme~\cite{Ad,Sy,Kos}. In terms of versatility, this goes beyond the results of \cite{CS2,CSV} which were confined to skew-symmetric classical $r$-matrices. The Lie dialgebra framework can easily accommodate the non-skew-symmetric case. For conciseness, we only illustrate this versatility and our construction on two famous models: the open Toda chain and the (rational) Gaudin model. However, the construction can in principle cover a much larger range of models which falls into the $r$-matrix scheme, see \cite{RSTS} for a description of many such systems including classical tops. To our knowledge, only one instance of a Lagrangian description of the AKS scheme has been proposed before in \cite{FG}. Compared to the present paper, \cite{FG} is limited to the AKS scheme and provides only one Lagrangian corresponding to the quadratic Hamiltonian $\Tr L^2/2$ (the idea of Lagrangian multiforms was not yet available at that time). By using ideas from Hamiltonian reduction, we produce a Lagrangian multiform on a general coadjoint orbit which encompasses the results of~\cite{FG} as a special case. 

Our main results are:
\begin{enumerate}
    \item The definition \eqref{our_Lag}-\eqref{pot_part} of a Lagrangian multiform from the data of a Lie dialgebra and the proof that its multi-time Euler-Lagrange equations produce a hierarchy of compatible equations in Lax form, Theorem \ref{Th_multi_EL}.

    \item For this Lagrangian multiform, the derivation of an identity relating its closure relation, its Euler-Lagrange equations and the Poisson involutivity of associated Hamiltonians, Theorem \ref{prop_double_zero}. 

\item The construction of a Lagrangian multiform from the reduction of a ``free'' Lagrangian on the cotangent bundle $T^*A$ of a Lie group $A$ and the connection with the above Lie dialgebra case.

    \item Explicit Lagrangian multiforms for the open Toda chain and the rational Gaudin model. 
\end{enumerate} 

The paper is organised as follows. In Section \ref{background}, we briefly review the notions of Lagrangian multiforms and Lie dialgebras that we need. Section \ref{gen_res} introduces the Lagrangian multiform and contains two main results, Theorems \ref{Th_multi_EL} and \ref{prop_double_zero}. Section \ref{reduction} deals with another main result. We recast our results in the context of reduction from free motion on the cotangent bundle of a Lie group and produce a Lagrangian multiform on a general coadjoint orbit. We show how to recover the case of a Lagrangian multiform associated with a Lie dialgebra described in Section \ref{gen_res}. In Section \ref{Flaschka}, we illustrate the construction for the open Toda chain associated with a Lie dialgebra via a non-skew-symmetric $r$-matrix. We present explicit expressions for the Lagrangian coefficients and relate our results to the well-known formulations of the Toda chain in Flaschka and canonical coordinates. In Section \ref{Toda_pq}, the same open Toda chain is used to illustrate our construction in the case of a skew-symmetric $r$-matrix. We also relate our results to the description in Flaschka and canonical coordinates. Section \ref{Gaudin} is concerned with the rational Gaudin model and is the opportunity for us to show how our Lagrangian multiform operates in the case of an infinite-dimensional Lie algebra which accounts for the presence of a spectral parameter in the Lax matrices. Although it deals with a finite-dimensional Gaudin model, this section bears a lot of similarities with the framework introduced in \cite{CSV} for integrable field theories. We end with concluding remarks in Section \ref{ccl}.

\section{Background material}\label{background}

\subsection[Lagrangian \texorpdfstring{$1$-forms}{one-form}]{Lagrangian $\boldsymbol{1}$-forms }

We review in more details the notion of Lagrangian multiforms that we need, restricting our attention to Lagrangian $1$-forms since our aim is to describe integrable hierarchies of finite-dimensional systems. The basic object is a Lagrangian 1-form 
\begin{equation}
\label{def_Lag_multiform}
    \Lag[q]=\sum_{k=1}^N \, \Lag_k[q] \, dt_k
\end{equation}
and the related generalised action
\begin{equation}
    S[q,\Gamma]=\int_\Gamma \Lag[q]
\end{equation}
where $\Gamma$ is a curve in the multi-time $\RR^N$ with (time) coordinates $t_1,\dots,t_N$ and $q$ denotes generic configuration coordinates. For instance, $q$ could be a position vector in $\RR^d$ for some $d$, or as will be the case for us, an element of a (matrix) Lie group. The notations $\Lag[q]$ and $\Lag_k[q]$ mean that these quantities depends on $q$ and a finite number of derivatives of $q$ with respect to the times $t_1,\dots,t_N$. In this paper, we restrict ourselves only to the case of first derivatives and simply write $\Lag_k$ for the Lagrangian coefficients.
The application of the generalised variational principle leads to the following multi-time Euler-Lagrange equations \cite{Su1}
\begin{align}
    \label{simple_multitime_EL1}
  &\frac{\partial \Lag_k}{\partial q}-\partial_{t_k}  \frac{\partial \Lag_k}{\partial q_{t_k}}=0\,,\\[1ex]
  \label{simple_multitime_EL2}
  &\frac{\partial \Lag_k}{\partial q_{t_\ell}}=0\,,\qquad \ell\neq k\,,\\[1ex]
  \label{simple_multitime_EL3}
  &\frac{\partial \Lag_k}{\partial q_{t_k}}=\frac{\partial \Lag_\ell}{\partial q_{t_\ell}}\,,\qquad k,\ell=1,\dots,N\,.
\end{align}
Note that \eqref{simple_multitime_EL1} is simply the standard Euler-Lagrange equation for each $\Lag_k$. Condition \eqref{simple_multitime_EL2} states that the Lagrangian coefficient $\Lag_k$ cannot depend on the velocities $q_{t_\ell}$ for $\ell\neq k$. The last condition \eqref{simple_multitime_EL3} requires that the conjugate momentum to $q$ be the same with respect to all times $t_k$. 
The closure relation then stipulates that 
\begin{equation}
\label{closure}
    d\Lag[q]=0~~\Leftrightarrow ~~\partial_{t_k}\Lag_j-\partial_{t_j}\Lag_k=0\,,
\end{equation}
on solutions of \eqref{simple_multitime_EL1}-\eqref{simple_multitime_EL3}.

\subsection{Lie dialgebras and Lax equations}

Here we collect facts from the theory of Lie dialgebras as defined in \cite[Lecture 2]{STS}, see also \cite[Chapter 4]{BBT}. Proofs are omitted for brevity and the reader is referred to \cite{STS,BBT} for details. We emphasise that Lie dialgebras are different from the perhaps more familiar Lie bialgebras appearing in Drinfeld's theory of Poisson-Lie groups. Connections and differences between these two structures are discussed in \cite{STS} and \cite{KS}.

Let $\g$ be a matrix Lie algebra, with matrix Lie group $G$, and $\g^*$ its dual space. We have the usual (co)adjoint actions\footnote{For simplicity, we only work with matrix Lie algebras and corresponding Lie groups.} for all $\xi\in\g^*$, $X,Y\in\g$, $g\in G$,
\begin{align}
 {\rm ad}_X\cdot Y&=[X,Y]\,,\qquad ({\rm ad}^*_X\cdot \xi)(Y)=-\xi({\rm ad}_X\cdot Y)=-\xi([X,Y])\,,\\[1ex]
 \text{Ad}_{g}\cdot X&=g\,X\,g^{-1}\,,\qquad   \text{Ad}^*_g\cdot \xi(X)=\xi(\text{Ad}_{g^{-1}}\cdot X)\,.
\end{align}
The space $\g^*$ can be endowed with the Lie-Poisson bracket defined by
\begin{equation}
\label{def_Lie_Poisson}
    \{f,g\}(\xi)= (\xi,[\nabla f(\xi)\,,\,\nabla g(\xi)])\,,~~f,g\in C^\infty(\g^*)\,,
\end{equation}
where we introduced the convenient notation $(\,~,~)$ for the natural pairing between $\g^*$ and $\g$: $\xi(X)=(\xi,X)$. The gradient $\nabla f(\xi)$ is the element of $\g$ defined from the differential $\delta f(\xi)$ by using the pairing
\begin{equation}
    \delta f(\xi)(\eta)=\lim_{\epsilon\to 0}\frac{f(\xi+\epsilon \eta)-f(\xi)}{\epsilon}=(\eta,\nabla f(\xi))\,.
\end{equation}
\begin{comment}
    Introducing a basis $\{E_\alpha\}$ of $\g$ and the dual basis $\{E^*_\alpha\}$ of $\g^*$ and coordinates functions $\xi_\alpha$ on $\g^*$ we find $\nabla f(\xi)=\frac{\partial f}{\partial \xi_\alpha}E_\alpha$ and the well-known coordinate form of the Lie-Poisson bracket
\begin{equation}
    \{f,g\}(\xi)=C_{\alpha\beta}^\gamma\, \xi_\gamma \,\frac{\partial f}{\partial \xi_\alpha}\frac{\partial g}{\partial \xi_\beta}
\end{equation}
where $C_{\alpha\beta}^\gamma$ are the structure constants of $\g$. 
\end{comment}
The Lie-Poisson bracket is degenerate in general and the $\text{Ad}^*$-invariant functions on $\g^*$ are the Casimir functions. Its symplectic leaves are the coadjoint orbits of $G$ in $\g^*$. The restriction to a coadjoint orbit gives rise to the Lie-Kostant-Kirillov-Souriau symplectic form $\omega_{KK}$. 

Let $R:\g\to\g$ be a linear map. It is a solution of the {\it modified} classical Yang-Baxter equation (mCYBE) if it satisfies
\begin{equation}
\label{mCYBE}
[R(X),R(Y)]-R\left([R(X),Y]+[X,R(Y)]\right)=-[X,Y]\,,~~\forall~X,Y\in\g\,.
\end{equation}
By abuse of language, we will call a solution $R$ of \eqref{mCYBE} a (classical) $r$-matrix,  in relation to the fact that with $R$ one can associate $r\in\g\otimes \g$ (which is what is traditionally called the $r$-matrix) when $\g$ is equipped with a nondegenerate ad-invariant symmetric bilinear form $\langle\,~,~\rangle$ ({\it e.g.} the Killing form when $\g$ is a finite-dimensional semi-simple Lie algebra).
A famous example of an $r$-matrix arises in the case where $\g$ admits a direct sum decomposition (as a vector space) into two Lie subalgebras 
\begin{equation}
	\g=\g_+\oplus\g_-\,.
\end{equation}
Then, $R=P_+-P_-$ is a solution of \eqref{mCYBE}, where $P_\pm$ is the projector on $\g_\pm$ along $\g_\mp$. 

Given a solution $R$ of the mCYBE, one can define on the vector space $\g$ a second Lie bracket
\begin{eqnarray}
\label{def_R_bracket}
	[X,Y]_R=\frac{1}{2}\left([R(X),Y]+[X,R(Y)]\right)\,.
\end{eqnarray}
The corresponding Lie algebra is denoted by $\g_R$. We therefore have an adjoint action of $\g_R$ on itself and a coadjoint action of $\g_R$ on $\g^*$ ($\g$ and $\g_R$, being the same vector space, have the same dual space)
\begin{align}
&	{\rm \text{ad}}^R_X\cdot Y=[X,Y]_R\,,\qquad\forall~X,Y\in\g\,,\\[1ex] 
 &({\rm \text{ad}^*}^R_X\cdot \xi)(Y)=-(\xi\,,\,{\rm \text{ad}}^R_X\cdot Y)=-(\xi\,,\,[X,Y]_R)\,.
\end{align}

The algebraic significance of the mCYBE and of the second Lie bracket $[\,~,~]_R$ is given by the following results which lead to essential factorisation properties underlying integrable systems. The key objects are the maps
\begin{equation}
    R_\pm=\frac{1}{2}\left(R\pm {\rm id}\right)\,.
\end{equation}
\begin{proposition}
	Let $\g_\pm={\rm Im}\,R_\pm$. Then,
	\begin{enumerate}
		\item $R_\pm:\g_R\to\g$ are Lie algebra homomorphisms:
		\begin{equation}
			R_\pm\left([X,Y]_R\right)=\left[R_\pm(X),R_\pm(Y)\right]\,.
		\end{equation} 
  In particular, $\g_\pm\subset\g$ are Lie subalgebras of $\g$.
	
	\item The mapping $i_R:\g_R\to \g_+\oplus\g_-$, $i_R(X)=(R_+(X),R_-(X))$ is a Lie algebra embedding. Thus $\widetilde{\g}_R={\rm Im}\,i_R$ is a Lie subalgebra of $\g_+\oplus \g_-$\,.

 \item The composition of the maps
 \begin{equation}
   i_R:  \g_R ~{\to}~ \g_+\oplus \g_-\,,~~X\mapsto(R_+(X),R_-(X))\,,
 \end{equation}
 followed by
  \begin{equation}
a:\g_+\oplus \g_-~{\to}~\g\,,~~(X_+,X_-)\mapsto X_+-X_-\,,
 \end{equation}
provides a unique decomposition of any element $X\in\g$ as $X=R_+(X)-R_-(X)$.
	\end{enumerate}
\end{proposition}
\noindent
Note that $R_+-R_-={\rm id}$ and
\begin{equation}
[X,Y]_R=	R_+\left([X,Y]_R\right)-R_-\left([X,Y]_R\right)=\left[R_+(X),R_+(Y)\right]-\left[R_-(X),R_-(Y)\right]\,.
\end{equation}
We can express the actions of $\g_R$ in terms of those of $\g$. For convenience, we write $X_\pm=R_\pm(X)$ for $X\in\g$. Then,
\begin{align}
&{\rm ad}^R_X\cdot Y=\frac{1}{2}{\rm ad}_{R(X)}\cdot Y+\frac{1}{2}{\rm ad}_{X}\cdot R(Y)={\rm ad}_{X_+}\cdot Y_+ -{\rm ad}_{X_-}\cdot Y_-\,,\\
\label{relation_ad_adR}
&{\rm \text{ad}^*}^R_X\cdot \xi=\frac{1}{2}{\rm \text{ad}}^*_{R(X)}\cdot\xi+\frac{1}{2}R^*({\rm \text{ad}}^*_{X}\cdot\xi)=R_+^*({\rm \text{ad}}^*_{X_+}\cdot\xi)-R_-^*({\rm \text{ad}}^*_{X_-}\cdot\xi)\,,
\end{align}
where the adjoint $A^*:\g^*\to\g^*$ of a linear map $A:\g\to\g$ is defined by $(A^*(\xi),X)=(\xi,A(X))$.

The application of this framework to integrable systems hinges on the interplay between the two Lie-Poisson brackets one can define on $\g^*$. Indeed, having a second Lie bracket, we can repeat the definition \eqref{def_Lie_Poisson} to obtain
\begin{equation}
\label{Lie_PB}
    \{f,g\}_R(\xi)= (\xi,[\nabla f(\xi)\,,\,\nabla g(\xi)]_R)\,.
\end{equation}
A similar conclusion holds: the symplectic leaves are coadjoint orbits of $G_R$, the Lie group of $\g_R$, in $\g^*$. The restriction to a coadjoint orbit gives rise to the symplectic form which we denote by $\omega_R$. 
It is the interplay between these two structures that provides integrable systems whose equations of motion take the form of a Lax equation. For this last part, one needs one more ingredient: an $\text{Ad}$-invariant nondegenerate bilinear symmetric form $\langle\,~,~\rangle$ on $\g$. It allows to identify $\g^*$ with $\g$ and the coadjoint actions with the adjoint actions.
Specifically, one has 
\begin{theorem}
The $\textup{Ad}^*\!$-invariant functions on $\g^*$ are in involution with respect to $\{\,~,~\}_R$. The equation of motion 
\begin{equation}
    \frac{d}{dt}L=\{L,H\}_R
\end{equation}
induced by an $\textup{Ad}^*\!$-invariant function $H$ on $\g^*$ takes the following equivalent forms, for an arbitrary $L\in\g^*$,
\begin{equation}
\label{Lax2}
   \frac{d}{dt}L=\textup{ad}^{R*}_{\nabla H(L)}\cdot L=\frac{1}{2}\,\textup{ad}^*_{R\nabla H(L)}\cdot L=\textup{ad}^*_{R_\pm \nabla H(L)}\cdot L\,.
\end{equation}
When there is an $\textup{Ad}$-invariant nondegenerate bilinear form $\langle\,~,~\rangle$ on $\g$ so that we can identify $\g^*$ with $\g$ and $\textup{ad}^*\!$ with $\textup{ad}$, the last equation takes the desired form of a Lax equation for $L\in\g$,
\begin{equation}
  \frac{d}{dt}L=[M_\pm,L]\,,\qquad M_\pm=R_\pm \nabla H(L)\,.
\end{equation}
\end{theorem}
The proof can be found for instance in \cite{BBT} and we only elaborate on certain points which will be useful for our purposes below. The crucial point is to exploit the $\text{Ad}^*$-invariance of the function $H$ defining the time flow. The latter means 
 that the following property holds
\begin{equation} 
\label{ad_invariance_property}
    \text{ad}^*_{\nabla H(\xi)} \cdot\xi = 0 ~~\Leftrightarrow ~~ (\xi\,,\, [\nabla H(\xi),X])=0\qquad \forall \xi \in \mathfrak{g}^*\,,~~\forall X\in\g\,.  
\end{equation}
Thus, for any two $\text{Ad}^*$-invariant functions $H_1$ and $H_2$,
\begin{align}
\label{involutivity}
    \{H_1,H_2\}_R(\xi)&= (\xi,[\nabla H_1(\xi)\,,\,\nabla H_2(\xi)]_R)\\
    &=\frac{1}{2}(\xi,[R\nabla H_1(\xi)\,,\,\nabla H_2(\xi)]+[\nabla H_1(\xi)\,,\,R\nabla H_2(\xi)])=0\,.
\end{align}
For any function $f$ on $\g^*$, the time evolution associated with the $\text{Ad}^*$-invariant $H$ with respect to the Poisson bracket $\{\,~,~\}_R$ is defined by
\begin{align*}
    \frac{d}{dt}f(L)=\{f,H\}_R(L) \,, 
\end{align*}
\ie
\begin{align*}
    \left(\frac{d}{dt}L\,,\, \nabla f(L)  \right)&= (L,[\nabla f(L)\,,\,\nabla H(L)]_R)=-\frac{1}{2}\,(L,[R\nabla H(L)\,,\,\nabla f(L)])\\
    &=(\text{ad}^{R*}_{ \nabla H(L)}\cdot L\,,\,\nabla f(L))=\frac{1}{2}(\text{ad}^{*}_{R\nabla H(L)}\cdot L\,,\,\nabla f(L))\,.
\end{align*}
Finally, in view of \eqref{def_R_bracket} and \eqref{ad_invariance_property}, we have 
\begin{equation}
    \text{ad}^{*}_{R\nabla H(L)}\cdot L=2\,\text{ad}^{*}_{R_\pm \nabla H(L)}\cdot L\,,
\end{equation}
thus establishing the various equivalent forms of the equations in \eqref{Lax2} (by restricting $f$ to be any of the coordinate functions on $\g^*$).

The involutivity property \eqref{involutivity} ensures that we can define compatible time flows associated with a family of $\text{Ad}^*$-invariant Hamiltonian functions $H_k$, $k=1,\dots,N$.
If one can supply enough such independent functions, or work on a coadjoint orbit of low-enough dimension, one obtains an integrable system described by an integrable hierarchy of equations in Lax form (again using the identification provided by $\langle\,~,~\rangle$)
\begin{equation}
\label{Lax_system}
    \partial_{t_k}L=[R_\pm \nabla H_k(L),L]\,,~~ k=1,\dots,N\,.
\end{equation}
The typical example of an invariant function $H_k$ is given by $H_k=\frac{1}{k+1}\Tr(L^k)$.

For our purposes, the Lie groups associated with $\g$ and $\g_R$ will be important. We introduce $G$ and $G_R$ as the (connected, simply connected) Lie groups defined for $\g$ and $\g_R$ respectively. For simplicity, we only think of matrix groups in this paper. Only in special circumstances are $G$ and $G_R$ diffeomorphic. In general, this is only true in a neighbourhood of the identity where the crucial difference between the two groups lies in their multiplications induced by $[\,~,~]$ and $[\,~,~]_R$ respectively. The homomorphisms $R_\pm$ give rise to Lie group homomorphisms (which we denote by the same symbols) and we obtain a factorisation at the group level. With $g=\text{e}^X$, $X\in\g$, we have 
\begin{equation}
    R_\pm\, g=\text{e}^{R_\pm X}\,.
\end{equation}
Specifically, let $G_\pm=R_\pm(G_R)$ be the subgroups of $G$ corresponding to $\g_\pm$. The composition of the maps 
 \begin{equation}
   i_R:  G_R\,\to\,G_+ \times G_-\,,~~ g\mapsto(R_+(g),R_-(g))\,,
 \end{equation}
 followed by
  \begin{equation}
m:G_+\times G_- \,\to\,G\,,~~ (g_+,g_-)\mapsto g_+ \, g_-^{-1}\,,
 \end{equation}
allows us to factorise {\it uniquely} an arbitrary element $g\in G$ (sufficiently close to the identity) as 
\begin{equation}
    g=g_+ \, g_-^{-1}\,,~~(g_+,g_-)\in \widetilde{G}_R={\rm Im}\,i_R\,.
\end{equation}
An element $g\in G_R$ can be identified with its image $(g_+,g_-)\in \widetilde{G}_R\subseteq G_+\times G_-$ and the multiplication $\cdot_R$ in $G_R$ is most easily visualised using the homomorphism property
\begin{equation}
    i_R( g\cdot_R h)=i_R( g)*i_R( h)=(g_+\,h_+\,,\,g_-\,h_-)
\end{equation}
where $*$ is the direct product group structure of $G_+\times G_-$. This is usually shortened to
\begin{equation}
  g\cdot_R h= (g_+\,h_+\,,\,g_-\,h_-)\,.
\end{equation}
The group $G_R$ acts on $\g_R$ by the adjoint action and on $\g^*$ via the coadjoint action
\begin{align}
     &\text{Ad}^R_g\cdot X =g\cdot_R X\cdot_R g^{-1}\,,\qquad  \forall \, X\in\g_R\,,~~g\in G_R\,,\\[1ex]
    &\text{Ad}^{R*}_g\cdot \xi(X) =(\xi\,,\,\text{Ad}^R_{g^{-1}}\cdot X)\,,\qquad  \forall \,g\in G_R\,,\xi\in\g^*\,,~~X\in\g_R\,.
\end{align}
\begin{remark}
\label{rem_assoc}
    When writing using the suggestive notation $g\cdot_R X\cdot_R g^{-1}$ for the adjoint action, we tacitly view $\cdot_R$ as an associative product on the matrix Lie algebra and its Lie group. Strictly speaking, this is not possible if $R$ is a solution of \eqref{mCYBE}. It becomes possible for instance if $\g$ is an associative algebra and we require $R$ to be a solution of the associative Yang-Baxter equation $R(X)\, R(Y)-R(R(X)\, Y+ X\, R(Y))+ X\, Y=0$, see \cite{STS2}. This implies that $X\cdot_R Y=\frac{1}{2}\left(R(X)\, Y+X\, R(Y) \right)$ defines a second associative product on $\g$ and allows us to view $[X,Y]_R$ as the commutator $X\cdot_R Y-Y\cdot_R X$, in complete analogy with $[X,Y]=X\, Y-Y\, X$. We will assume that $\cdot_R$ is such an associative product in the rest of this paper and use the consequences, {\it e.g.} $[X,Y]_R=X\cdot_R Y-Y\cdot_R X$.
\end{remark}
The following relations are most useful in the practical calculations of the examples discussed below. With $g_\pm=R_\pm\,g$, $X_\pm=R_\pm\,X$, $g\in G_R$, $X\in \g_R$,
\begin{align}
   \text{Ad}^R_g\cdot X&=g_+\, X_+\,g_+^{-1}-g_-\, X_-\,g_-^{-1}\,,\\[1ex]
\label{Rcoadjoint_action} 
\text{Ad}^{R*}_g\cdot \xi&=R_+^*(\text{Ad}^*_{g_+} \xi)-R_-^*(\text{Ad}^*_{g_-}\xi)\,,~~  \forall \xi\in\g^*\,.
\end{align}
Thus, the dual space $\g^*$ hosts two coadjoint actions of $G$ and $G_R$, as it does with the two coadjoint actions of the Lie algebras $\g$ and $\g_R$. 
The last main result of this framework is known as the {\it factorisation theorem}, see {\it e.g.} \cite{STS,BBT,AVV}.
\begin{theorem}
 Consider the system of compatible equations with the given initial condition 
\begin{equation}
\label{system_Lax}
 \partial_{t_k}L=\textup{ad}^*_{R_\pm \nabla H_k(L)}\cdot L\,,\qquad k=1,\dots,N\,,
\qquad L(0,\dots,0)=L_0\in\g^*\,.
\end{equation} 
Denote $(t_1,\dots,t_N)={\bf t}$ for conciseness. Let $g_\pm({\bf t})$ be the smooth curves in $G_\pm$ which solves the factorisation problem
    \begin{equation}
    \label{factorisation}
     \text{e}^{-\sum_{k=1}^N t_k\nabla H_k(L_0)}=g_+({\bf t})^{-1}\,g_-({\bf t})\,,\qquad g_\pm({\bf 0})=e\,.
    \end{equation}
    Then, the solution to the initial-value problem \eqref{system_Lax} is given by
    \begin{equation}
    \label{solution_L}
        L({\bf t})=\textup{Ad}^*_{g_+({\bf t})}\cdot L_0=\textup{Ad}^*_{g_-({\bf t})}\cdot L_0\,,
    \end{equation}
    and $g_{\pm}({\bf t})$ satisfy 
     \begin{equation}
          \label{split}
  \partial_{t_k} g_\pm({\bf t}) =R_\pm \nabla H_k(L({\bf t}))\,g_\pm({\bf t})\,.
          \end{equation}
\end{theorem}
This result shows that the solution lies at the intersection of coadjoint orbits of $G$ and $G_R$. Combined with the fact that the coadjoint orbits provide the natural symplectic manifolds associated with the corresponding Lie-Poisson bracket, this means that the natural arena to define our phase space, \ie where $L$ lives, is a coadjoint orbit of $G_R$ in $\g^*$
\begin{equation} \label{eq:coadj_orbit}
    {\cal O}_\Lambda=\{\text{Ad}^{R*}_{\varphi}\cdot \Lambda;\varphi\in G_R\}\,,~~\text{for some}~~\Lambda\in\g^*\,. 
\end{equation}
In the Lagrangian multiform theory, the prevalent idea is that one should think of an integrable system as an integrable hierarchy, in a way completely similar to the Hamiltonian integrable hierarchy we have just recalled. This leads us to work with the space where $(t_1,\dots,t_N)={\bf t}$ lives: the {\it multi-time} space. Since the flows commute, the multi-time is simply (a subspace of) $\mathbb{R}^{N_1}\times (S^1)^{N_2}$, $N_1+N_2=N$ (in general we should allow for the possibility of having periodicity in some of the independent variables $t_1,\dots,t_N$). The generalisation to the case where the vector fields giving the flows no longer commute but still form a Lie algebra was considered in \cite{CNSV} and leads to the consideration of the multi-time space being a (non-abelian) Lie group.

\subsection{Extension to loop algebras and some special cases}\label{subs:special_cases_AKS_gaudin}

The essential results of the Lie dialgebra construction discussed above extend to the infinite-dimensional setting, {\it e.g.} the case of loop algebras\footnote{There are several subtleties related to duals in infinite dimensions and completions which we do not touch, keeping a less rigorous but more approachable exposition.}. The latter is relevant when one needs Lax matrices with spectral parameters. This is typically the case for integrable field theories but it can also be required for some finite-dimensional systems such as the closed Toda chain or Gaudin models. We will present the extension of the Lie dialgebra construction to this infinite-dimensional setting via the Gaudin example in Section \ref{Gaudin} and we refer the reader to \cite[Lecture~3]{STS} for more details. 

There are special cases of the Lie dialgebra framework that may be more familiar to the reader and will play a role in our examples below. They both arise when $\g$ admits a direct sum decomposition (as a vector space) into two Lie subalgebras
\begin{equation}
	\g=\g_+\oplus\g_-\,,
\end{equation}
and we take $R=P_+-P_-$, where $P_\pm$ is the projector on $\g_\pm$ along $\g_\mp$. The decomposition of $\g$ induces the decomposition 
\begin{equation}
	\g^*=\g_+^*\oplus\g_-^*\,.
\end{equation}
Using a nondegenerate ad-invariant bilinear form on $\g$, we can identify $\g_\pm^*$ with $\g_\mp^\perp$.

The first special case, which historically is at the origin of the so-called Adler-Kostant-Symes scheme \cite{Ad,Kos,Sy} is obtained as follows. We fix $\Lambda$ to be in $\g_-^*$ and consider the coadjoint orbit of elements $L=\text{Ad}^{R*}_{\varphi} \cdot \Lambda$. As a result, only the subgroup $G_-$ in $G_R\simeq G_+\times G_-$ plays a role since $L=\text{Ad}^{R*}_{\varphi} \cdot \Lambda=-R_-^*(\text{Ad}^*_{\varphi_-}\cdot \Lambda)$ and the coadjoint orbit ${\cal O}_\Lambda$ lies in $\g_-^*$. This is the historic setup which can be used to formulate the open Toda chain in Flaschka coordinates. $R$ is {\it not} skew-symmetric in this case. We will present this example in Section \ref{Flaschka} where details on our Lagrangian multiform for this model will be given.

The second special case is a further specialisation where $\g_\pm$ are isotropic with respect to $\langle\,~,~\rangle$, meaning
$$\langle\g_\pm,\g_\pm\rangle=0$$
and implying that $\g_\pm^*$ can be identified with $\g_\mp=\g_\mp^\perp$. This case can arise with loop algebras and will be discussed in Section \ref{Gaudin} in relation to the Gaudin model. In this case, $R$ is skew-symmetric, \ie
$$\langle RX,Y\rangle=-\langle X,RY\rangle\,,~~\forall~X,Y\in\g\,.$$

Note that we will also illustrate the case where $R$ is not defined from a decomposition into two subalgebras but rather from a decomposition into nilpotent and Cartan subalgebras. This different setup is accommodated without problems into Lie dialgebras. Interestingly, it can also be used to describe the same open Toda chain as in the AKS scheme and this will be illustrated in Section \ref{Toda_pq}. The underlying algebraic structures are very different, though. In particular, $R$ is skew-symmetric in this case while it is not in the AKS formulation, showing that the same Toda chain can arise from two distinct constructions.

\section{Lagrangian multiform on a coadjoint orbit}\label{gen_res}

\subsection{The Lagrangian multiform and its properties}\label{gen_lm}

Recalling our comment about the coadjoint orbits of $G_R$ in $\g^*$ being the natural arena for an integrable hierarchy, let us introduce the following Lagrangian $1$-form
\begin{equation}
\label{our_Lag}
  \Lag[\varphi] = \sum_{k=1}^N \, \Lag_k \, dt_k ={\cal K}[\varphi]-{\cal H}[\varphi]
\end{equation}
with kinetic part
\begin{equation}
\label{kin_part}
  {\cal K}[\varphi] =  \sum_{k=1}^N \, \left(\, L\,,\,\partial_{t_k}\varphi \cdot_R \varphi^{-1}  \,\right) \, dt_k \,,~~   L = \text{Ad}^{R*}_{\varphi} \cdot \Lambda\,,~~\varphi \in G_R\,, 
\end{equation}
and potential part
\begin{equation}
\label{pot_part}
{\cal H}[\varphi]=\sum_{k=1}^N \, H_k(L) \, dt_k \,. 
\end{equation}
 The field $\varphi\in G_R$ contains the dynamical degrees of freedom and, as we will see, the Euler-Lagrange equation will take a natural form when expressed in terms of $L = \text{Ad}^{R*}_{\varphi} \cdot \Lambda$. $\Lambda$ is a fixed non-dynamical element of $\g^*$ which defines ${\cal O}_\Lambda$, the phase space of the model. Each Lagrangian $\Lag_k$ in the Lagrangian multiform has a structure comparable to the familiar Lagrangian $p\dot{q}-H$ in classical mechanics.
The potential part is expressed in terms of $\text{Ad}^*$-invariant functions $H_k \in C^{\infty}(\mathfrak{g}^*)$ and we suppose we have $N$ of them\footnote{At this stage, we do not necessarily have that $N$ is exactly half of the dimension of ${\cal O}_\Lambda$. As in the AKS scheme, this needs to be addressed in specific cases by choosing a coadjoint orbit of appropriate dimension to ensure Liouville integrability. We will not worry about this for now as our construction follows through anyway.}.
\begin{remark}
We emphasised that one important ingredient in producing equations of motion in Lax form from the coadjoint orbit construction is to use an $\text{Ad}$-invariant nondegenerate bilinear symmetric form $\langle\,~,~\rangle$ on $\g$ to identify $\g^*$ with $\g$ and the coadjoint action with the adjoint action. The reader could therefore wonder why we have written our Lagrangian multiform using the pairing $(\,~,~)$, an element $\Lambda\in\g^*$ and functions $H_k$ on $\g^*$. The point is that we found that it was less confusing to do so when deriving results in general and in examples, in order to identify correctly the subalgebras involved in the decomposition of $\g$ and $\g^*$. However, we cannot stress enough that ultimately we always use the bilinear form $\langle\,~,~\rangle$ to make all the identifications and indeed obtain equations in Lax form, whether this is clearly mentioned or not. Hopefully, this understanding will make the exposition easier to follow.        
\end{remark}        
We can now formulate our first main result.
\begin{theorem}
\label{Th_multi_EL}
The Lagrangian $1$-form \eqref{our_Lag} satisfies the corner equations \eqref{simple_multitime_EL2}-\eqref{simple_multitime_EL3} of the multi-time Euler-Lagrange equations. The standard Euler-Lagrange equations \eqref{simple_multitime_EL1} associated with the Lagrangian coefficients $\Lag_k$ take the form of compatible Lax equations
        \begin{equation}
\label{EL_Lax}
    \partial_{t_k}L=[R_\pm \nabla H_k(L),L]\,,\qquad k=1,\dots,N\,.
\end{equation}
The closure relation holds: on solutions of \eqref{EL_Lax} we have 
$$ \partial_{t_k}\Lag_j-\partial_{t_j}\Lag_k=0\,,\qquad j,k=1,\dots,N\,.$$
\end{theorem}
\begin{proof} It is clear that each $\Lag_k$ does not depend on $\partial_{t_\ell} \varphi$ for $\ell\neq k$ so the corner equation \eqref{simple_multitime_EL2} is satisfied. To see that \eqref{simple_multitime_EL3} holds, it is convenient to introduce local coordinates $\phi_\alpha$, $\alpha=1,\dots,M$, on the group $G_R$. The only source of dependence on velocities is in the kinetic term of $\Lag_k$. Now 
\begin{equation} \label{kinetic_in_coords}
\begin{split}
  \left(\, \text{Ad}^{R*}_{\varphi} \cdot \Lambda\,,\,\partial_{t_k}\varphi \cdot_R \varphi^{-1}  \,\right)
  &=\left(\,  \Lambda\,,\,\text{Ad}^{R}_{\varphi^{-1}}\cdot \left(\partial_{t_k}\varphi \cdot_R \varphi^{-1}\right)\right) \\[1ex]
&=  \left(\,  \Lambda\,,\,\varphi^{-1} \cdot_R  \partial_{t_k}\varphi \right) \\[1ex]  
&= \sum_{\alpha=1}^M\left(\,  \Lambda\,,\,\varphi^{-1} \cdot_R  \frac{\partial \varphi}{\partial \phi_\alpha}    \,\right)\partial_{t_k}\phi_\alpha \equiv \sum_{\alpha=1}^M  \pi_\alpha\, \partial_{t_k}\phi_\alpha 
\end{split} 
\end{equation}
where we have introduced the momentum
\begin{equation}
    \pi_\alpha=\left(\,  \Lambda\,,\,\varphi^{-1} \cdot_R  \frac{\partial \varphi}{\partial \phi_\alpha}    \,\right)
\end{equation}
conjugate to the field $\phi_\alpha$. Thus, 
$$\frac{\partial \Lag_k}{\partial \left(\partial_{t_k}\phi_\alpha\right)}=\pi_\alpha$$
is independent of $k$.
The remainder of the multi-time Euler Lagrange equations consists of the standard Euler-Lagrange equations for each $\Lag_k$. We compute 
\begin{equation}
\begin{split}
    \delta \Lag_k  &= \left(\, \delta L\,,\,\partial_{t_k}\varphi \cdot_R \varphi^{-1}  \,\right) +  \left(\, L\,,\,\delta\,(\partial_{t_k}\varphi \cdot_R \varphi^{-1})  \,\right)-\delta H_k(L)  \,,
\end{split}
\end{equation}
with\footnote{More rigorously, the notation $\delta L$ means the tangent vector to ${\cal O}_\Lambda$ at the point $L$ induced by the element $X\in\g_R$ which we write more suggestively as $\delta \varphi\cdot_R\varphi^{-1}$. The latter notation is closer to the more familiar one in variational calculus using matrix-valued fields.}
\begin{equation}
    \delta L=\text{ad}^{R*}_{\delta \varphi\cdot_R\varphi^{-1}}\cdot L
\end{equation}
and
\begin{align*}
\delta H_k(L)&=\left(\delta L,\nabla H_k(L) \right) =-\left( L\,,\,\left[ \delta \varphi\cdot_R\varphi^{-1},\nabla H_k(L)\right]_R \right)\\[1ex]
&=\frac{1}{2}\left( L\,,\,\left[R\nabla  H_k(L)\,,\, \delta \varphi\cdot_R\varphi^{-1}\right] \right)=-\frac{1}{2}\left(\text{ad}^*_{R\nabla  H_k(L)}\cdot L\,,\,\delta \varphi\cdot_R\varphi^{-1} \right)\,.
\end{align*}
So, 
\begin{align*}
    \delta \Lag_k &= \left(\, \text{ad}^{R*}_{\delta \varphi\cdot_R\varphi^{-1}}\cdot L\,,\,\partial_{t_k}\varphi \cdot_R \varphi^{-1}  \,\right) +  \left(\, L\,,\,\delta\,(\partial_{t_k}\varphi) \cdot_R \varphi^{-1}  \,\right)\\[1ex] 
    &~~- \left(\, L\,,\,\partial_{t_k}\varphi \cdot_R \varphi^{-1} \cdot_R \delta \varphi\cdot_R \varphi^{-1}  \,\right)+\frac{1}{2}\left(\text{ad}^*_{RdH_k(L)}\cdot L\,,\,\delta \varphi\cdot_R\varphi^{-1} \right) \\[1ex]
    &= \left(\, \text{ad}^{R*}_{\delta \varphi\cdot_R\varphi^{-1}}\cdot L\,,\,\partial_{t_k}\varphi \cdot_R \varphi^{-1}  \,\right) - \left(\, \partial_{t_k}L\,,\,\delta\,\varphi \cdot_R \varphi^{-1} \,\right) + \left(\, L\,,\,\delta \varphi\cdot_R \varphi^{-1} \cdot_R \partial_{t_k}\varphi \cdot_R \varphi^{-1}  \,\right)\\[1ex]
    &~~ + \partial_{t_k}\left(\, L\,,\,\delta\,\varphi \cdot_R \varphi^{-1}  \,\right)- \left(\, L\,,\,\partial_{t_k}\varphi \cdot_R \varphi^{-1} \cdot_R \delta \varphi\cdot_R \varphi^{-1}  \,\right)+\frac{1}{2}\left(\text{ad}^*_{R\nabla H_k(L)}\cdot L\,,\,\delta \varphi\cdot_R\varphi^{-1} \right)\\[1ex]
    &= \left(\, \text{ad}^{R*}_{\delta \varphi\cdot_R\varphi^{-1}}\cdot L\,,\,\partial_{t_k}\varphi \cdot_R \varphi^{-1}  \,\right) - \left(\, \text{ad}^{R*}_{\partial_{t_k}\varphi\cdot_R\varphi^{-1}}\cdot L\,,\,\delta\,\varphi \cdot_R \varphi^{-1} \,\right) \\[1ex]
    &~~+ \left(\, L\,,\,\left[\delta \varphi\cdot_R \varphi^{-1} \,,\,\partial_{t_k}\varphi \cdot_R \varphi^{-1} \right]_R \,\right)+\frac{1}{2}\left(\text{ad}^*_{R\nabla H_k(L)}\cdot L\,,\,\delta \varphi\cdot_R\varphi^{-1} \right)+\partial_{t_k}\left(\, L\,,\,\delta\,\varphi \cdot_R \varphi^{-1}  \,\right).
\end{align*}
The first and third term cancel each other. In the second term we recognise $\text{ad}^{R*}_{\partial_{t_k}\varphi\cdot_R\varphi^{-1}}\cdot L=\partial_{t_k} L$. Hence, 
\begin{align*}
\label{delta_L}
    \delta \Lag_k=
    \left(\, -\partial_{t_k} L+\frac{1}{2}\text{ad}^*_{R\nabla H_k(L)}\cdot L\,,\,\delta\,\varphi \cdot_R \varphi^{-1} \,\right) +\partial_{t_k}\left(\, L\,,\,\delta\,\varphi \cdot_R \varphi^{-1}  \,\right)
\end{align*}
and we obtain the Euler-Lagrange equation for each $\Lag_k$ as 
\begin{equation}
  \partial_{t_k} L=  \frac{1}{2}\, \text{ad}^*_{R\nabla H_k(L)}\cdot L\,.
\end{equation}
Now recall that $  \frac{1}{2}\, \text{ad}^*_{R\nabla H_k(L)}\cdot L= \text{ad}^*_{R_\pm\nabla H_k(L)}\cdot L$ and that, with $\g$ being equipped with an $\text{Ad}$-invariant nondegenerate bilinear form, $\text{ad}^*_{R_\pm\nabla H_k(L)}\cdot L$ is identified with $[R_\pm\nabla H_k(L), L]$. Thus, we have obtained \eqref{EL_Lax} variationally as desired.
That this set of equations is compatible follows from the commutativity of the flows which is a consequence of the mCYBE and the $\text{Ad}$-invariance of $H_k$ as we now show. Going back to having $L\in\g^*$ and evaluating its derivatives on a fixed but arbitrary $X\in\g$, we have
\begin{align*}
    (\partial_{t_k}\partial_{t_j}L)(X)&=-\frac{1}{2}\partial_{t_k}\left( L\,,\, [R\nabla H_j(L),X] \right)\\
    &=\frac{1}{4}\left( L\,,\,[R\nabla H_k(L)\,,\, [R\nabla H_j(L)\,,\,X]] \right)-\frac{1}{4}\left( L\,,\,[R [R\nabla H_k(L)\,,\,\nabla H_j(L)]\,,\,X] \right)\,.
\end{align*}
Hence, using the Jacobi identity
\begin{align*}
    ([\partial_{t_k}\,,\,\partial_{t_j}]L)(X)&=\frac{1}{4}\left( L\,,\,[[R\nabla H_k(L)\,,\, R\nabla H_j(L)]\,,\,X]\right)\\
    &~~-\frac{1}{4}\left(L\,,\,R ([R\nabla H_k(L)\,,\,\nabla H_j(L)]+[\nabla H_k(L)\,,\,R\nabla H_j(L)])\,,\,X] \right)\\
    &=-\frac{1}{4}\left( L\,,\,[[\nabla H_k(L)\,,\, \nabla H_j(L)]\,,\,X]\right)=0
\end{align*}
where we use the mCYBE in the second equality and property \eqref{ad_invariance_property} in the last step.
We now establish the closure relation, \ie $d\Lag=0$ on shell. It turns out that the kinetic and potential contributions vanish separately. 
We have 
\begin{equation}
    \partial_{t_j} \Lag_k-\partial_{t_k} \Lag_j = \partial_{t_j} \left(\,L\,,\,\partial_{t_k}\varphi \cdot_R \varphi^{-1}  \,\right)-\partial_{t_k} \left(\,L\,,\,\partial_{t_j}\varphi \cdot_R \varphi^{-1}  \,\right) -\partial_{t_j}H_k(L)+\partial_{t_k}H_j(L)\,.
\end{equation}
Now, using \eqref{ad_invariance_property}, we find
\begin{equation}\label{dtjhkl}
\partial_{t_j}H_k(L)=\left(\partial_{t_j} L \,,\, \nabla H_k(L)\right)=-\frac{1}{2}\,\left( L \,,\, \left[R\nabla H_j(L)\,,\,\nabla H_k(L)\right]\right)=0.
\end{equation}
Thus, it is a direct consequence of the $\text{Ad}^*$-invariance of $H$ that the potential contribution to $d\Lag$ is zero on shell. We are now left with just the kinetic terms which can be rewritten as
\begin{align*}
  &\left(\,\partial_{t_j} L\,,\,\partial_{t_k}\varphi \cdot_R \varphi^{-1}  \,\right)- \left(\,\partial_{t_k} L\,,\,\partial_{t_j}\varphi \cdot_R \varphi^{-1}  \,\right)+ \left(\, L\,,\,\partial_{t_j}(\partial_{t_k}\varphi \cdot_R \varphi^{-1} ) \,\right)-\left(\, L\,,\,\partial_{t_k}(\partial_{t_j}\varphi \cdot_R \varphi^{-1} ) \,\right)\\
    &= \left(\,\partial_{t_j} L\,,\,\partial_{t_k}\varphi \cdot_R \varphi^{-1}  \,\right)- \left(\,\partial_{t_k} L\,,\,\partial_{t_j}\varphi \cdot_R \varphi^{-1}  \,\right)+ \left(\, L\,,\,\partial_{t_j}\partial_{t_k}\varphi \cdot_R \varphi^{-1}\,- \,\partial_{t_k}\partial_{t_j}\varphi \cdot_R \varphi^{-1}\,\right)\\
    &~~+ \left(\, L\,,\,\partial_{t_k}\varphi \cdot_R \partial_{t_j}\varphi^{-1}\,-\,\partial_{t_j}\varphi \cdot_R \partial_{t_k}\varphi^{-1}\,\right).
\end{align*}
From the commutativity of flows, we have $\partial_{t_j}\partial_{t_k}\varphi - \partial_{t_k}\partial_{t_j}\varphi = 0$, which leaves us with
\begin{eqnarray*}
\left(\,\partial_{t_j} L\,,\,\partial_{t_k}\varphi \cdot_R \varphi^{-1}  \,\right)- \left(\,\partial_{t_k} L\,,\,\partial_{t_j}\varphi \cdot_R \varphi^{-1}  \,\right)+ \left(\, L\,,\,\partial_{t_k}\varphi \cdot_R \partial_{t_j}\varphi^{-1}\,-\,\partial_{t_j}\varphi \cdot_R \partial_{t_k}\varphi^{-1}\,\right).
\end{eqnarray*}
The on-shell relation
$$\partial_{t_j} L= \frac{1}{2}\, \text{ad}^*_{R\nabla H_j(L)}\cdot L\,$$
allows us to express the first term as
\begin{align*}
\left(\,\partial_{t_j} L\,,\,\partial_{t_k}\varphi \cdot_R \varphi^{-1}\,\right) &= \frac{1}{2}  \left(\,\text{ad}^*_{R\,\nabla H_j(L)}\cdot L\,,\,\partial_{t_k}\varphi \cdot_R \varphi^{-1}\,\right)\\
&= -\frac{1}{2}  \left(\,\text{ad}^*_{\partial_{t_k}\varphi \cdot_R \varphi^{-1}}\cdot L\,,R\,\nabla H_j(L)\right).
\end{align*}
Since
$$\left(\,\text{ad}^{R*}_{\partial_{t_k}\varphi \cdot_R \varphi^{-1}}\cdot L\,,\nabla H_j(L)\right) = \frac{1}{2} \left(\,\text{ad}^*_{\partial_{t_k}\varphi \cdot_R \varphi^{-1}}\cdot L\,,R\nabla H_j(L)\right) + \frac{1}{2} \left(\,\text{ad}^*_{R\,\partial_{t_k}\varphi \cdot_R \varphi^{-1}}\cdot L\,,\nabla H_j(L)\right)$$
and
$$\left(\,\text{ad}^*_{R\,\partial_{t_k}\varphi \cdot_R \varphi^{-1}}\cdot L\,,\nabla H_j(L)\right)=-\left(\,\text{ad}^*_{\nabla H_j(L)}\cdot L\,,R\partial_{t_k}\varphi \cdot_R \varphi^{-1}\right) = 0\,,$$
we have a further simplification to
\begin{align*}
    \left(\,\partial_{t_j} L\,,\,\partial_{t_k}\varphi \cdot_R \varphi^{-1}\,\right) &= -  \left(\,\text{ad}^{R*}_{\partial_{t_k}\varphi \cdot_R \varphi^{-1}}\cdot L\,,\nabla H_j(L)\right)\\[1ex]
    &= - \left(\,\partial_{t_k}L\,,\nabla H_j(L)\right)\\[1ex]
    &= -\,\partial_{t_k}H_j(L) = 0
\end{align*}
where we have used the result from \eqref{dtjhkl} (with $k\leftrightarrow j$).
Similarly, we have for the second term
\begin{equation}
     \left(\,\partial_{t_k} L\,,\,\partial_{t_j}\varphi \cdot_R \varphi^{-1}\,\right) = -\,\partial_{t_j}H_k(L) = 0.
\end{equation}
For the last remaining term, we have
\begin{align*}
    &\left(\, L\,,\,\partial_{t_k}\varphi \cdot_R \partial_{t_j}\varphi^{-1}\,-\,\partial_{t_j}\varphi \cdot_R \partial_{t_k}\varphi^{-1}\,\right) \\[1ex]
   &\qquad\qquad=  \left(\, L\,,\,-\partial_{t_k}\varphi \cdot_R \varphi^{-1} \cdot_R \partial_{t_j}\varphi \cdot_R \varphi^{-1}\,+\,\partial_{t_j}\varphi \cdot_R \varphi^{-1} \cdot_R \partial_{t_k}\varphi \cdot_R \varphi^{-1}\,\right)\\[1ex]
     &\qquad\qquad= \left(\, L\,,\, \left[ \,\partial_{t_j}\varphi \cdot_R \varphi^{-1}\,,\, \partial_{t_k}\varphi \cdot_R \varphi^{-1}\, \right]_R\,\right)\\[1ex]
     &\qquad\qquad= -\left(\,\text{ad}^{R*}_{\partial_{t_j}\varphi \cdot_R \varphi^{-1}}\cdot L\,,\,\partial_{t_k}\varphi \cdot_R \varphi^{-1}\,\right)\\[1ex]
     &\qquad\qquad= - \left(\,\partial_{t_j} L\,,\,\partial_{t_k}\varphi \cdot_R \varphi^{-1}\,\right)= \,\partial_{t_k}H_j(L) = 0.
\end{align*}
\end{proof}
It is worth noting that the properties of our Lagrangian multiform heavily rely on the mCYBE for $R$. It is at the heart of the commutativity of the flows and the closure relation. The connection between the closure relation and the CYBE was first identified and established in \cite{CSV} in the context of integrable field theories. Here it is established in the finite-dimensional context and related to Lie dialgebras.

\subsection{Closure relation, Hamiltonians in involution, Kostant-Kirillov form}\label{closure_KK}

In this section, we derive a structural result which brings together Lagrangian multiforms and essential Hamiltonian aspects of integrable systems. It will be convenient and clearer to work with local coordinates $\phi_\alpha$, $\alpha=1,\dots,M$, on the group $G_R$, as we did in \eqref{kinetic_in_coords}. Then, our Lagrangian multiform can be written in the form 
\begin{equation} 
\label{eq:kin_term_canonical_coord}
  \Lag[\varphi] =\sum_{k=1}^N\left( \,\sum_{\alpha=1}^M  \pi_\alpha\, \partial_{t_k}\phi_\alpha - H_k\right) dt_k\\  
\end{equation}
where we recall that the momentum $\pi_\alpha$ is defined by
\begin{equation}
    \pi_\alpha=\left(\,  \Lambda\,,\,\varphi^{-1} \cdot_R  \frac{\partial \varphi}{\partial \phi_\alpha}    \,\right)\,.
\end{equation}
Each Lagrangian $\Lag_k$ in the multiform has the structure $p\dot{q}-H$ of a Lagrangian in phase space, 
\begin{equation}
    \Lag_k=\sum_{\alpha=1}^M  \pi_\alpha\, \partial_{t_k}\phi_\alpha - H_k\,,
\end{equation}
and yields its Euler-Lagrange equations from the variation
\begin{equation}
   \delta \Lag_k=\sum_{\beta=1}^M\left(\sum_{\alpha=1}^M \left(\frac{\partial \pi_\alpha}{\partial \varphi_\beta}-\frac{\partial \pi_\beta}{\partial \varphi_\alpha} \right)\, \partial_{t_k}\phi_\alpha - \frac{\partial H_k}{\partial \phi_\beta}\right)\delta\phi_\beta+\partial_{t_k}\left(\sum_{\alpha=1}^M \pi_\alpha\delta \phi_\alpha \right)\,.
\end{equation}
This is of course consistent with the general result of the previous section, and the comparison of the two expressions for $\delta \Lag_k$ gives 
\begin{equation}
\label{link}
\left(\, -\partial_{t_k} L+\frac{1}{2}\,\text{ad}^*_{R\nabla H_k(L)}\cdot L\,,\,\delta\,\varphi \cdot_R \varphi^{-1} \,\right)=    \sum_{\beta=1}^M\left(\sum_{\alpha=1}^M \Omega_{\alpha\beta}\, \partial_{t_k}\phi_\alpha - \frac{\partial H_k}{\partial \phi_\beta}\right)\delta\phi_\beta
\end{equation}
and 
\begin{equation}
\left(\, L\,,\,\delta\,\varphi \cdot_R \varphi^{-1}  \,\right)=\left(\sum_{\alpha=1}^M \pi_\alpha\delta \phi_\alpha \right)\,.
\end{equation}
Thus, we have natural coordinate versions of key components of the theory. In particular, let us denote by $\theta_R$ the vertical $1$-form
\begin{equation}
    \theta_R=-\sum_{\alpha=1}^M \pi_\alpha\delta \phi_\alpha=-\sum_{\alpha=1}^M \left(\,  \Lambda\,,\,\varphi^{-1} \cdot_R  \frac{\partial \varphi}{\partial \phi_\alpha}    \,\right)\delta \phi_\alpha=-\left(\,  \Lambda\,,\,\varphi^{-1} \cdot_R  \delta \varphi\,\right)\,,
\end{equation}
and let us introduce the vertical $2$-form
\begin{equation}
    \Omega_R=\sum_{\alpha<\beta} \Omega_{\alpha\beta}\,\delta \phi_\alpha\wedge \delta \phi_\beta\,,\qquad \Omega_{\alpha\beta}=\frac{\partial \pi_\alpha}{\partial \phi_\beta}-\frac{\partial \pi_\beta}{\partial \phi_\alpha}\,.
\end{equation}
Observe the important relation
\begin{equation}
\label{exact_form}
    \Omega_R=\delta \theta_R\,.
\end{equation}
The form $\Omega_R$ is the pullback to the group $G_R$ by the map
\begin{equation}
\begin{split}
    \chi:~&G_R \to {\cal O}_\Lambda \\[1ex] 
    &\varphi \mapsto \text{Ad}^{R*}_{\varphi} \cdot \Lambda
\end{split} 
\end{equation}
of the Kostant-Kirillov symplectic form $\omega_{R}$ on the coadjoint orbit through $\Lambda\in\g^*$. We recall here that we consider the coadjoint action of the group $G_R$, not the group $G$. 
Relation \eqref{exact_form} is the well-known fact that this pullback is an exact form. The expression $\varphi^{-1} \cdot_R  \delta \varphi$ appearing in $\theta_R$ can be interpreted as the Maurer-Cartan form on $G_R$. The structure of our Lagrangian coefficients, in particular their kinetic part, is now elucidated in terms of fundamental objects associated with $G_R$ and its coadjoint orbits in $\g^*$. 

It is known that the map $\chi$ is a submersion\footnote{We suppose that we are in a situation where this holds, for instance, excluding the trivial case where the orbit is reduced to a point and assuming that the $G_R$ action is proper.}. Also, a coadjoint orbit is always even-dimensional as it admits the nondegenerate symplectic form $\omega_{R}$. Let us introduce local coordinates $\xi_m$, $m=1,\dots, 2p$, on ${\cal O}_\Lambda$ ($2p\le M$). The tangent map $\chi_*$ is represented locally by the $2p\times M$ matrix $\left(\frac{\partial \xi_m}{\partial \phi_\alpha}\right)$. From now on, summation over repeated indices is understood. The pushforward of the vector fields $\frac{\partial }{\partial \phi_\alpha}$ on $G_R$
is given by
\begin{equation}
    \chi_*\!\left(\frac{\partial }{\partial \phi_\alpha}\right) =\frac{\partial \xi_m}{\partial \phi_\alpha}\,\frac{\partial }{\partial \xi_m}
\end{equation}
and the pullback of the differential $1$-forms $\delta \xi_m$ on ${\cal O}_\Lambda$ reads
\begin{equation}
    \chi^*(\delta \xi_m)=\frac{\partial \xi_m}{\partial \phi_\alpha}\, \delta \phi_\alpha\,.
\end{equation}
If we write for the Kostant-Kirillov form
\begin{equation}
    \omega_R=\omega_{mn}\,\delta \xi_m\wedge \delta \xi_n \,, 
\end{equation}
then we have the following relation with the coefficients of its pullback $\Omega_R=\chi^*(\omega_R)$,
\begin{equation}
    \Omega_{\alpha\beta}= \frac{\partial \xi_m}{\partial \phi_\alpha}\,\frac{\partial \xi_n}{\partial \phi_\beta}\,  \omega_{mn} \,. 
\end{equation}
In view of \eqref{link}, it remains to introduce the Euler-Lagrange vertical $1$-forms on $G_R$
\be
EL_k\equiv EL_k^\beta\,\delta \phi_\beta \equiv\left(\Omega_{\alpha\beta}\,\partial_{t_k}\phi_\alpha -\frac{\partial  H_k}{\partial \phi_\beta}\right)\,\delta \phi_\beta\,.
\ee
This is the pullback of the following vertical $1$-form on ${\cal O}_\Lambda$,
\be
EL_k=\chi_*(\Upsilon_k)= \Upsilon_k^n\,\,\chi_*(\delta \xi_n)=\left(\sum_{m}\omega_{mn}\, \partial_{t_k}\xi_m -\frac{\partial  H_k}{\partial \xi_n} \right)\chi_*(\delta \xi_n)
\ee
with the relation
\be
EL_k^\beta= \Upsilon_k^n\,\frac{\partial \xi_n}{\partial \phi_\beta}\,.
\ee
Since $\chi$ is a submersion, the matrix $\left(\frac{\partial \xi_m}{\partial \phi_\alpha}\right)$ has maximal rank $2p$, so the Euler-Lagrange equations $EL_k^\beta=0$ imply the equations $\Upsilon_k^n=0$ (and vice versa). This is of course just the confirmation in the present coordinate notations of the result we obtained previously that the (multi-time) Euler-Lagrange equations from our Lagrangian multiform produce Lax equations naturally living on coadjoint orbits of $G_R$. 

As a consequence, whenever we say that an equality holds ``on shell'', we mean that it holds modulo $EL_k^\beta=0$ or equivalently $\Upsilon_k^n=0$. We can take advantage of this in the following way. $\Omega_R$ is the pullback of the Kostant-Kirillov form $\omega_R$ on the coadjoint orbit ${\cal O}_R$. The latter is nondegenerate and therefore induces a Poisson bracket 
with bivector
\begin{equation}
    P_R=\sum_{m<n}P_{mn}\frac{\partial }{\partial \xi_m}\wedge \frac{\partial }{\partial \xi_n}\,,\qquad P_{mn}\, \omega_{nr}=\delta_{mr}\,.
\end{equation}
The corresponding Poisson bracket on ${\cal O}_\Lambda$ is known (see {\it e.g.} \cite[Chapter 14]{BBT}) to be the restriction of the Lie-Poisson bracket \eqref{Lie_PB} on $\g^*$ 
\begin{equation}
    \{f,g\}_R(\xi)=\Big(\xi,\left[\nabla f(\xi)\,,\nabla g(\xi)\right]_R\Big)\,.
\end{equation}
In other words, when $f$, $g$ are restricted to ${\cal O}_\Lambda$, we have 
\be
\label{def_PB}
\{f,g\}_R=P_{mn}\frac{\partial f}{\partial \xi_m}\frac{\partial g}{\partial \xi_n}\,.
\ee
With these notions introduced, we see that the Euler-Lagrange equations $\Upsilon_k^n=0$ take the form
\begin{equation}
\label{ELk}    \sum_{m}\omega_{mn}\,\partial_{t_k}\xi_m =\frac{\partial  H_k}{\partial \xi_n}\,,
\end{equation}
and can be written in Hamiltonian form 
\be
\label{syst_eqs_q}
\partial_{t_k}\xi_m=P_{mn}\frac{\partial H_k}{\partial \xi_n}=\{\xi_m,H_k\}_R\,.
\ee
The system of simultaneous equations \eqref{syst_eqs_q} on the $\xi_m$ admits a solution (at least locally) if and only if the flows are compatible, \ie $[\partial_{t_k},\partial_{t_\ell}]=0$. For an arbitrary function $f$, this means 
$$[\partial_{t_k},\partial_{t_\ell}]f=\{\{H_k, H_\ell\}_R,f\}_R=0\,.$$
The stronger condition $\{ H_k, H_\ell\}_R=0$ is the familiar Hamiltonian criterion for integrability (together with a sufficient number of independent such functions $H_k$, of course).

After these preliminary steps, we are now ready to state our second main result and its corollary, the significance of which will be discussed after the proofs.
\begin{theorem}\label{prop_double_zero}
The following {\bf identity} holds
	\be
\frac{\partial \Lag_k}{\partial t_\ell}-\frac{\partial \Lag_\ell}{\partial t_k}+\Upsilon_k^m\,P_{mn}\,\Upsilon_\ell^n=\{H_k,H_\ell\}_R\,.
	\ee    
\end{theorem}
\begin{proof}
    The proof is by direct computation.
    \begin{align*}
        \frac{\partial \Lag_k}{\partial t_\ell}-\frac{\partial \Lag_\ell}{\partial t_k}&=\left(\frac{\partial \pi_\alpha}{\partial \phi_\beta}-\frac{\partial \pi_\beta}{\partial \phi_\alpha}\right)\partial_{t_\ell}\phi_\beta\,\partial_{t_k} \phi_\alpha-\frac{\partial H_k}{\partial \phi_\beta}\partial_{t_\ell}\phi_\beta+\frac{\partial H_\ell}{\partial\phi_\alpha}\partial_{t_k}\phi_\alpha\\
    &=\left(\Omega_{\alpha\beta} \,\partial_{t_k} \phi_\alpha-\frac{\partial H_k}{\partial \phi_\beta}\right)\partial_{t_\ell}\phi_\beta+\frac{\partial H_\ell}{\partial \xi_m}\frac{\partial \xi_m}{\partial \phi_\alpha}\partial_{t_k}\phi_\alpha\\
    &=\left(\omega_{mn} \,\partial_{t_k} \xi_m-\frac{\partial H_k}{\partial \xi_n}\right)\partial_{t_\ell}\xi_n+\frac{\partial H_\ell}{\partial \xi_m}\partial_{t_k}\xi_m\\
    &=\left(\omega_{mn} \,\partial_{t_k} \xi_m-\frac{\partial H_k}{\partial \xi_n}\right)P_{nr}\left( \omega_{rs}\,\partial_{t_\ell}\xi_s+\frac{\partial H_\ell}{\partial \xi_r}-\frac{\partial H_\ell}{\partial \xi_r}  \right)+\frac{\partial H_\ell}{\partial \xi_m}\partial_{t_k}\xi_m\\
    &=-\Upsilon_k^n\,P_{nr}\,\Upsilon_\ell^r+\frac{\partial H_k}{\partial \xi_n}\,P_{nr} \, \frac{\partial H_\ell}{\partial \xi_r},
    \end{align*}
    hence the result.
\end{proof}
\begin{corollary}
\label{coro}
The closure relation for the Lagrangian multiform $\Lag$ is equivalent to the involutivity of the Hamiltonians $H_k$ with respect to the Lie-Poisson $R$-bracket $\{\,~,~\}_R$.
\end{corollary}
\begin{proof}
    The closure relation 
    requires that {\it on shell}, we have 
    \be
d\Lag=\sum_{k<\ell}\left(\frac{\partial \Lag_\ell}{\partial t_k}-\frac{\partial \Lag_k}{\partial t_\ell}\right)dt_k\wedge dt_\ell=0\,.
\ee
From the previous theorem, on shell we have 
	\be
\frac{\partial \Lag_k}{\partial t_\ell}-\frac{\partial \Lag_\ell}{\partial t_k}=\{H_k,H_\ell\}_R\,,
	\ee    
hence the result.
\end{proof}
\begin{remark}
    The connection between the closure relation for Lagrangian $1$-forms and the involutivity of Hamiltonians was first discussed in \cite{Su1}. The content of our Corollary establishes this result for all Lagrangian $1$-forms in the class that we have introduced in this paper. They include any system describable by the coadjoint orbit and $r$-matrix methods of Lie dialgebras. An extension of the connection between closure and involutivity to the field theory context (Lagrangian $2$-forms) was discussed in \cite{Ver2}. 
\end{remark}
\begin{remark}
    In the field theory context, the connection between the closure relation and the classical Yang-Baxter equation was elucidated in \cite{CSV}. In the present article on Lagrangian $1$-forms, this connection is also at the heart of our results since the entire construction is based on the availability of the second Lie bracket $[\,~,~]_R$ on $\g$, a feature ensured if $R$ satisfies the mCYBE. 
\end{remark}
\begin{remark}
    The content of the theorem sheds fundamental light on the link between the closure relation and the involutivity of the Hamiltonian as it establishes an {\it off-shell} identity which clearly shows the interplay between the coefficients of $d\Lag$, the Euler-Lagrange equations, the Poisson tensor on the coadjoint orbit, and the Poisson bracket of the Hamiltonians related to our Lagrangian coefficients. A particular point is that it shows in the present general setting that $d\Lag$ has a so-called ``double zero'' on the equation of motion. This idea was introduced in \cite{SNC2} and developed in \cite{SNC3,S1} as an important ingredient of Lagrangian multiform theory. However, the relation to Hamiltonians in involution was not noticed there. The status of the ``double zero'' term $\Upsilon_k^n\,P_{nr}\,\Upsilon_\ell^r$ is now clearly identified as well as its relation to the Euler-Lagrange equations. This term is the off-shell element linking the Hamiltonian integrability criterion $\{H_k,H_\ell\}_R=0$ and the integrability criterion advocated in Lagrangian multiform theory: the closure relation $d\Lag=0$ on shell. 
\end{remark}

\section{Lagrangian multiform on a coadjoint orbit from reduction}\label{reduction}

It is well-known that many integrable systems arise from Hamiltonian reduction on the cotangent bundle of a Lie group $A$ following the intuitive idea that the more intricate dynamics of the integrable system of interest on the reduced phase space comes from the simplest ``free'' dynamics on the cotangent bundle. In this section, we show how one can construct a general Lagrangian multiform on a coadjoint orbit by a Lagrangian analogue of the procedure of Hamiltonian reduction. The multiform of Section \ref{gen_res} is recovered as a special case. 

We follow mainly the exposition and ideas in \cite[Lectures 1 $\&$ 2]{STS} to summarise the key notions. By fixing a left trivialisation of $T^*A$ we can parametrise it with $(\alpha,a)\in\mathfrak{a}^*\times A$ where $\mathfrak{a}^*$ is the dual of the Lie algebra $\mathfrak{a}$ of $A$. The canonical symplectic form $\Omega$ is exact and derives from the canonical one-form $\theta$
\begin{equation}
    \Omega=\delta \theta~~\text{with}~~\theta=(\alpha,a^{-1}\delta a)\,.
\end{equation}
The cotangent lifts to $T^*A$ of the action of $A$ on itself by left and right translations read
\begin{equation}
    \lambda_b:(\alpha,a)\mapsto(\alpha,b\,a)\,,\quad \rho_b:(\alpha,a)\mapsto(\text{Ad}_b^* \cdot\alpha\,,a\,b^{-1})\,,~~b\in A\,.
\end{equation}
The canonical one-form and hence the symplectic form are invariant under these actions. The corresponding moment maps are given by
\begin{equation}
    \mu_{\ell}(\alpha,a)=\text{Ad}^*_{a}\cdot \alpha \,,~~\mu_{r}(\alpha,a)=-\alpha\,.
\end{equation}
In applications to integrable systems, one usually consider the case where only Lie subgroups $A_+$ and $A_-$ of $A$ act by left and right translations. In this case, the moment maps are the restriction of the above moment maps to $\mathfrak{a}_\pm$, the Lie algebras of $A_\pm$. Thus they are elements of $\mathfrak{a}^*_\pm$ and we denote them by
\begin{equation}
\label{projected_moments}
\mu_\ell(\alpha,a)=\Pi_{\mathfrak{a}^*_+}\left(\text{Ad}^*_{a}\cdot \alpha \right) \,,~~\mu_{r}(\alpha,a)=-\Pi_{\mathfrak{a}^*_-}\alpha\,.
\end{equation}
In the special case where $A_+$ is the trivial group and $A_-=A$, it is known that the quotient Poisson manifold $T^*A/A$ is isomorphic to $\mathfrak{a}^*$ equipped with the Lie-Poisson bracket, see {\it e.g.} \cite[Proposition 1.24]{STS}.

Since our emphasis is on the Lagrangian formalism, let us describe the translation of the above situation into this framework. We consider the following Lagrangian on $T^*A$ 
\begin{equation}
    \Lag^0=\left(\alpha,a^{-1}\frac{da}{dt} \right)\,.
\end{equation}
The importance of $\Lag^0$ is that the Cartan form arising from its variation is precisely the canonical one-form on $T^*A$. Indeed, we have
\begin{align*}
    \delta \Lag^0&=\left(\delta \alpha,a^{-1}\frac{da}{dt} \right)-\left(\alpha,a^{-1}\delta a a^{-1}\frac{da}{dt} \right)+\left( \alpha,a^{-1}\delta\frac{da}{dt} \right)\\
    &=\left(\delta \alpha,a^{-1}\frac{da}{dt} \right)-\left(\alpha,a^{-1}\delta a a^{-1}\frac{da}{dt} \right)-\left(\frac{d\alpha}{dt},a^{-1}\delta a \right)+\left(\alpha,a^{-1}\frac{da}{dt} a^{-1}\delta a  \right)+\frac{d}{dt}\left(\alpha,a^{-1}\delta a \right)\\
    &=\left(\delta \alpha,a^{-1}\frac{da}{dt} \right)-\left(\frac{d}{dt}\left(\text{Ad}^*_a\cdot \alpha\right),\delta a a^{-1}\right)+\frac{d}{dt}\left(\alpha,a^{-1}\delta a \right)
\end{align*}
In the last term we recognise that the Cartan form is $\theta$ (up to a conventional sign). Also, we see that this Lagrangian yields trivial (free) equations of motion
\begin{equation}
\label{free_eqs}
    a^{-1}\frac{da}{dt}=0\,,~~\frac{d}{dt}\left(\text{Ad}^*_a\cdot \alpha\right)=0\quad \Leftrightarrow \quad  \frac{da}{dt}=0\,,~~\frac{d\alpha}{dt}=0\,.
\end{equation}
The Lagrangian $\Lag^0$ is invariant under the global transformations $(\alpha,a)\mapsto(\alpha,b\,a)$ and $(\alpha,a)\mapsto(\text{Ad}_b^*\cdot \alpha,a\,b^{-1})$ where $b\in A$ is constant. The conserved currents produced by Noether's theorem are the moment maps $\mu_{\ell,r}$. It is immediate from \eqref{free_eqs} that they are indeed conserved currents. The symmetry group $A\times A$ of this free theory is too large to produce systems of interest. One easy way to reduce the symmetry group to $A_+\times A_-=\{e\}\times A$ acting by right translations only is to include a potential term where the potential function depends on $(\alpha,a)$ only through $\mu_\ell$:
\begin{equation}
    \Lag=\left(\alpha,a^{-1}\frac{da}{dt} \right)-H\left(-\mu_\ell(\alpha,a)\right)=\left(\alpha,a^{-1}\frac{da}{dt} \right)-H\left(\text{Ad}^*_{a}\cdot \alpha\right)\,,
\end{equation}
where $H$ is a function on $\mathfrak{a}^*$. By Noether's theorem, we expect that $\mu_r=-\alpha$ is still a conserved current. Indeed, a computation analogous to that in the proof of Theorem \eqref{Th_multi_EL} gives 
\begin{align}
    \delta \Lag
    &=\left(\delta \alpha,a^{-1}\frac{da}{dt} -\text{Ad}_{a^{-1}}\cdot\nabla H\left(\text{Ad}^*_{a}\cdot \alpha\right)\right)\nonumber\\
    &-\left(\frac{d}{dt}\left(\text{Ad}^*_a\cdot \alpha \right)-\text{ad}^*_{\nabla H\left(\text{Ad}^*_{a}\cdot \alpha\right)}\cdot \left(\text{Ad}^*_{a}\cdot\alpha\right),\delta a a^{-1}\right)+\frac{d}{dt}\left(\alpha,a^{-1}\delta a \right)
\end{align}
Thus the equations of motion read
\begin{equation}
    \frac{da}{dt}\,a^{-1} -\nabla H\left(\text{Ad}^*_{a}\cdot \alpha\right)=0\,,~~
\frac{d}{dt}\left(\text{Ad}^*_a\cdot \alpha \right)-\text{ad}^*_{\nabla H\left(\text{Ad}^*_{a}\cdot \alpha\right)}\cdot \left(\text{Ad}^*_{a}\cdot\alpha\right)=0
\end{equation}
or equivalently
\begin{equation}
    \frac{da}{dt}\,a^{-1} -\nabla H\left(\text{Ad}^*_{a}\cdot \alpha\right)=0\,,~~
\frac{d}{dt}\alpha=0\,.
\end{equation}
The analogue of fixing the moment map $\mu_r=-\alpha$ to some fixed value $-\Lambda\in\mathfrak{a}^*$ in the Hamiltonian reduction approach consists of ``integrating out degrees of freedom'' by solving $\frac{d}{dt}\alpha=0$ to $\alpha=\Lambda\in\mathfrak{a}^*$, and inserting back into the Lagrangian to get the effective Lagrangian of the reduced model. This yields 
\begin{equation}
\label{eff_Lag}
    \Lag_{\text{eff}}=\left(\Lambda,a^{-1}\frac{da}{dt} \right)-H\left(\text{Ad}^*_{a}\cdot \Lambda\right)\,.
\end{equation}
This Lagrangian describe a system on the coadjoint orbit of $\Lambda\in\mathfrak{a}^*$ under $A$. At this stage, if we equip $\mathfrak{a}$ with a nondegenerate symmetric bilinear form to identify $\mathfrak{a}^*$ with $\mathfrak{a}$, we obtain as before that the equations of motion take the Lax form for $L=\text{Ad}_{a}^*\cdot \Lambda$
\begin{equation}
    \frac{dL}{dt}=\left[\nabla H(L)\,,\, L\right]\,.
\end{equation}
Note that we did not assume anything special about the function $H$ so that strictly speaking there is no notion of integrability at this stage, only that the equations for the system under consideration are written in Lax form. 

Applying this construction to the case $A=G_R$, $\mathfrak{a}=\g_R$ we see that each of our Lagrangian $\Lag_k$ in \eqref{our_Lag} is of the form of $\Lag_{\text{eff}}$. Of course, in that case, each function $H_k$ was assumed to have the additional property of being invariant under the coadjoint action of $G$ so that the closure relation, the Lagrangian criterion for integrability\footnote{Modulo, as always, the requirement of having enough such independent functions.}, was valid. However, let us go back to the general situation above and suppose we form a Lagrangian 1-form by assembling $N$ effective Lagrangians of the form \eqref{eff_Lag} with $N$ independent (arbitrary smooth) functions $H_k$ defined on $\mathfrak{a}^*$ (or possibly only on ${\cal O}_\Lambda$):
\begin{equation}
\label{general_form_A}
  \Lag = \sum_{k=1}^N \, \Lag_k \, dt_k = \sum_{k=1}^N \, \left(\left(\Lambda,a^{-1}\partial_{t_k}a \right)-H_k\left(\text{Ad}^*_{a}\cdot \Lambda\right) \right)\, dt_k\,.
\end{equation}
The arguments of Section \ref{closure_KK} can be repeated verbatim and lead to the same conclusion as in Theorem \ref{prop_double_zero}.
\begin{theorem}
\label{gen_th}
The following {\bf identity} holds
	\be
 \label{general_off_shell}
\frac{\partial \Lag_k}{\partial t_\ell}-\frac{\partial \Lag_\ell}{\partial t_k}+\Upsilon_k^m\,P_{mn}\,\Upsilon_\ell^n=\{H_k,H_\ell\}_{\mathfrak{a}^*}\,,
	\ee    
 where $\{~,~\}_{\mathfrak{a}^*}$ is the Lie-Poisson bracket on $\mathfrak{a}^*$ and $P_{mn}$ the corresponding Poisson tensor on ${\cal O}_\Lambda$.
\end{theorem}
Note that so far we have not assumed anything on the functions $H_k$. The exact analogue of Corollary \ref{coro} follows from Theorem \ref{gen_th} in the present context. We stress its importance in this general situation: if we can solve for $\Lag_k$ such that the multi-time Euler-Lagrange equations and the closure relation hold for the $1$-form \eqref{general_form_A} then it qualifies as a Lagrangian multiform and \eqref{general_off_shell} implies that the corresponding functions $H_k$ are in involution. Conversely, if we use functions $H_k$ that are in involution with respect to $\{~\,,~\}_{\mathfrak{a}^*}$ then the 1-form \eqref{general_form_A} satisfies the closure relation and is a Lagrangian multiform. In Section \ref{gen_res}, we used the latter point of view in a special situation: we took advantage of the Lie dialgebra construction which uses $\text{Ad}^*_G$-invariant functions to produce functions in involution with respect to $\{~\,,~\}_R$ on the dual of $\g_R$. In the present section, the above results imply the stronger statement that one can associate a Lagrangian multiform with any family of Hamiltonians in involution on any coadjoint orbit of a Lie group. 

The perspective of the reverse procedure consisting of solving the multi-time Euler-Lagrange equations and the closure relations to produce (new?) integrable systems and Hamiltonians in involution is tantalising. However, it is far from clear whether such a philosophy is more (or less) promising than the established integrable system classification tools such as symmetry analysis or classification of the solutions of the (modified) classical Yang-Baxter equation.  

It is instructive to recover our Lagrangians $\Lag_k$ as effective Lagrangians via a slightly different, but related, mechanism. Suppose now that we reduce the symmetry group to $A_+\times A_-$ such that, at least locally, any element $a\in A$ factorise uniquely as $a=a_+^{-1}a_-$, $a_\pm\in A_\pm$. At the Lie algebra level, we have a unique decomposition $X=X_+-X_-$, $X_\pm\in \mathfrak{a}_\pm$ and by duality $\alpha=\Pi_{\mathfrak{a}^*_+}\alpha-\Pi_{\mathfrak{a}^*_-}\alpha$. Identifying $a$ with $(a_+,a_-)$, we see that the action of $A_+\times A_-$ on $a$ amounts to an action by right translations $(a_+,a_-)\mapsto(a_+b_+^{-1},a_-b_-^{-1})$. Thus,
\begin{equation}
    s:T^*A\to \mathfrak{a}^*\,,~~(\alpha,a)\mapsto \text{Ad}^*_{a_-}\cdot \alpha
\end{equation}
is invariant under the action of $A_+\times A_-$. Equipped with this, as before, it is easy to introduce a Lagrangian which has $A_+\times A_-$ as symmetry group 
\begin{equation}
    \Lag=\left(\alpha,a^{-1}\frac{da}{dt} \right)-H\left(-s(\alpha,a)\right)\,,
\end{equation}
where $H$ is a function on $\mathfrak{a}^*$. By Noether's theorem we expect that $\mu_{\ell,r}$ in \eqref{projected_moments} are conserved currents.
The direct verification from the Euler-Lagrange equations follows by noticing that 
$$\left(\alpha,a^{-1}\frac{da}{dt} \right)=-\left(s,\frac{da_+}{dt}a_+^{-1}-\frac{da_-}{dt}a_-^{-1} \right)\,.$$
Therefore 
\begin{equation}
\begin{aligned}
    \delta \Lag=&-\left(\delta s,\frac{da_+}{dt}a_+^{-1}-\frac{da_-}{dt} a_-^{-1}-\nabla H(-s)\right)+\left( \frac{d}{dt}\left( \text{Ad}^*_{a_+^{-1}}\cdot s\right)\,,\,\delta a_+a_+^{-1} \right)\\[1ex]
    &-\left( \frac{d}{dt}\left( \text{Ad}^*_{a_-^{-1}}\cdot s\right)\,,\,\delta a_-a_-^{-1} \right) - \frac{d}{dt}\left(s,\delta a_+a_+^{-1}-\delta a_-a_-^{-1} \right)
\end{aligned}
\end{equation} 
giving 
\begin{equation}
\begin{cases}
    \dfrac{d}{dt}\,\Pi_{\mathfrak{a}^*_+}\left(\text{Ad}^*_{a_+^{-1}}\cdot s \right)=\dfrac{d}{dt}\,\Pi_{\mathfrak{a}^*_+}\left(\text{Ad}^*_{a}\cdot \alpha \right)=\dfrac{d}{dt}\,\mu_\ell=0\,,\\[2ex]
    \dfrac{d}{dt}\,\Pi_{\mathfrak{a}^*_-}\left(\text{Ad}^*_{a_-^{-1}}\cdot s \right)=\dfrac{d}{dt}\,\Pi_{\mathfrak{a}^*_-} \alpha =-\dfrac{d}{dt}\,\mu_r=0\,,\\[2ex]
    \dfrac{da_+}{dt}\,a_+^{-1}-\dfrac{da_-}{dt}\,a_-^{-1} -\nabla H(-s)=0\,.
\end{cases}
\end{equation}
    The first two equations are indeed the conservation of the Noether currents, as expected. We use them to integrate out the corresponding degrees of freedom. Namely, we set 
    \begin{equation}
        \mu_\ell=-\Lambda_+\in\mathfrak{a}_+\,,~~\mu_r=\Lambda_-\in\mathfrak{a}_-,
    \end{equation}
    and substitute back into the Lagrangian. To obtain the effective Lagrangian, note that 
    \begin{equation}
        s=\Pi_{\mathfrak{a}^*_+}\left(\text{Ad}^*_{{a_+}}\cdot \mu_\ell \right)+\Pi_{\mathfrak{a}^*_-}\left(\text{Ad}^*_{{a_-}}\cdot \mu_r \right)\,.
    \end{equation}
    This can be seen by the following computation\footnote{This computation is the generalisation to the present context of the analogous argument used in \cite{FG} where our $s$ corresponds to their $X$.}. For any $X\in\mathfrak{a}$,
\begin{align*}
    \left(s,X\right)=\left(s,X_+-X_-\right)=&\left(\text{Ad}^*_{a_+^{-1}}\cdot s,\text{Ad}_{a_+^{-1}}\cdot X_+\right)-\left(\text{Ad}^*_{a_-^{-1}}\cdot s,\text{Ad}_{a_-^{-1}}\cdot X_-\right)\\
    =&\left(\mu_\ell,\text{Ad}_{a_+^{-1}}\cdot X_+\right)+\left(\mu_r,\text{Ad}_{a_-^{-1}}\cdot X_-\right)\\
=&\left(\Pi_{\mathfrak{a}^*_+}\!\left(\text{Ad}^*_{{a_+}}\cdot \mu_\ell \right)+\Pi_{\mathfrak{a}^*_-}\!\left(\text{Ad}^*_{{a_-}}\cdot \mu_r \right),X\right)\,.
\end{align*}
Putting everything together, and setting 
\begin{equation}
    L=-s=\Pi_{\mathfrak{a}^*_+}\left(\text{Ad}^*_{{a_+}}\cdot \Lambda_+ \right)-\Pi_{\mathfrak{a}^*_-}\left(\text{Ad}^*_{{a_-}}\cdot \Lambda_- \right)
\end{equation}
the effective Lagrangian is
\begin{equation}
    \Lag_{\text{eff}}=\left(L\,,\frac{da_+}{dt}\,a_+^{-1} -\frac{da_-}{dt}\,a_-^{-1}\right)-H\left(L\right)\,.
\end{equation}
In the special case where $A=G$, $\mathfrak{a}=\g$, $\mathfrak{a}_\pm=\g_{\pm}$ with $\g_\pm=R_\pm(\g)$, this effective Lagrangian is exactly of the form of our Lagrangian coefficients $\Lag_k$. This alternative construction amounts to reducing a free system on $T^*G$ by acting with $G_+\times G_-\simeq G_R$. We refer the interested reader to \cite[Section 2.4]{STS} for a discussion of the connection between the reduction on $T^*G_R$ by left translations of $G_R$ and the reduction on $T^*G$ by left and right translations of $G_+\times G_-$. The main reason why we discussed the alternative construction here is because it suggests that one may be able to construct Lagrangian multiforms for systems of Calogero-Moser type. In the terminology of \cite[Section 2.4]{STS}, one would have to drop the assumption that the subgroup $H$ of $G\times G$ used for reduction is transversal to the diagonal subgroup and implement instead a reduction corresponding for instance to the action of $G$ on itself by conjugation. We leave this open problem for future investigation.

\section{Open Toda chain in the AKS scheme}\label{Flaschka}

As this is our first example, we first spend some time reviewing the known Adler-Kostant-Symes Lie algebraic construction of the Lax matrix and the Lax equation reproducing Flaschka's approach. Then, we will make the connection with our variational approach.

	\paragraph{Algebraic setup:} Let us choose $\g=\sl(N+1)$, the Lie algebra of $(N+1)\times (N+1)$ traceless real matrices, $\g_+$ the Lie subalgebra of skew-symmetric matrices and $\g_-$ the Lie subalgebra of upper triangular traceless matrices, yielding
	\be
	\g=\g_+\oplus \g_-\,.
	\ee
	Here $R=P_+-P_-$ and $R_\pm=\pm P_\pm$ with $P_\pm$ the projector on $\g_\pm$ along $\g_\mp$. 
	The following $\text{Ad}$-invariant nondegenerate bilinear form
	\be
 \label{bilinear_form}
	\langle\,X,Y\rangle={\rm Tr}(XY)
	\ee
	allows the identification $\g^*\simeq \g$, and it induces the decomposition  
	\be
	\g^*=\g_-^*\oplus \g_+^*\simeq\g_+^\perp\oplus \g_-^\perp\,,
	\ee
	where $\g_\pm^\perp$ is the orthogonal complement of $\g_{\pm}$ with respect to $\langle\,~,~\rangle$: $\g_+^\perp$ is the subspace of traceless symmetric matrices and $\g_-^\perp$ the subspace of strictly upper triangular matrices. Let us choose
	\be \label{eq:Lambda}
	\Lambda =\begin{pmatrix}
		0 & 1 & 0 & 0 &\dots & 0\\
		1 & 0 & 1 & 0 &\dots & 0\\
		0 & 1 & 0 & 1 &\dots & 0\\
		0 & 0 & 1 &\ddots &\ddots & \vdots\\
		\vdots & & & \ddots & \ddots & 1\\
		0 & 0 & 0 & \dots &1 & 0\\
	\end{pmatrix}\in \g_-^*\simeq\g_+^\perp
	\ee
	and consider its orbit under the (co)adjoint action of $G_-$, the Lie subgroup associated to $\g_-$ consisting of upper triangular matrices with unit determinant. 
	
	\paragraph{Lax matrix and Lax equations for the first two flows:} As explained in Section \ref{subs:special_cases_AKS_gaudin}, the AKS case corresponds to the particular case where $\varphi\in G_-$ so that 
	\begin{equation*}
		L=\text{Ad}^{R*}_{\varphi} \cdot \Lambda=-R_-^*(\text{Ad}^*_{\varphi_-}\cdot \Lambda),
	\end{equation*}
	and the coadjoint orbit ${\cal O}_{\Lambda}$ lies in $\g_-^*$. 
	Using $\langle\,~,~\rangle$ we can identify the adjoint and coadjoint actions. Also, we use it to identify the transpose $A^*:\g^*\to\g^*$ of any linear map $A:\g\to\g$ with the transpose of $A$ with respect to $\langle\,~,~\rangle$ defined on $\g$. Writing $(\xi,X)=\langle Y,X\rangle$, this means that we have
	$$(A^*(\xi),X)=(\xi,A(X))=\langle Y,A(X) \rangle=\langle A^*(Y),X \rangle\,.$$
	This allows us to work with 
	\begin{equation}
		L=-R_-^*(\varphi_-\, \Lambda\,\varphi_-^{-1})=-R_-^*(\varphi\, \Lambda\,\varphi^{-1}) \,,
	\end{equation}
	where we have dropped the redundant subscript on $\varphi$ in the second equality with $\varphi=\varphi_-\in G_-$.
From the definitions $\langle\,X\,,\,R_\pm Y\,\rangle=\langle\,R^*_\pm X\,,\,Y\,\rangle$ and  $\langle\,X\,,\,P_\pm Y\,\rangle=\langle\,\Pi_\mp X\,,\,Y\,\rangle$, where we denote by $\Pi_\pm$ the projector onto $\g_\pm^\perp$ along $\g_\mp^\perp$, we find $R^*_\pm=\pm\Pi_\mp\,$. Note that this is an example of non-skew-symmetric $r$-matrix since
 \begin{equation}
     R^*=\Pi_--\Pi_+\neq -R=P_--P_+\,.
 \end{equation}
	Now, $\varphi\,\Lambda \,\varphi^{-1}$ is of the form
	\be
	\varphi\,\Lambda \,\varphi^{-1}=\begin{pmatrix}
		~a_1~ & * & * & * &\dots & *\\[1ex]
		b_1 & a_2 & * & * &\dots & *\\[1ex]
		0 & b_2 & a_3 & * &\dots & *\\[1ex]
		0 & 0 & b_3 &\ddots &\ddots & \vdots\\[1ex]
		\vdots & & & \ddots & \ddots & *\\[1ex]
		0 & 0 & 0 & \dots &b_{N} & ~a_{N+1}~\\
	\end{pmatrix}\,.
	\ee
So, we find
	\be
	\label{form_L_Flaschka}
	L=\Pi_+(\varphi\,\Lambda \,\varphi^{-1})=\begin{pmatrix}
		~a_1~ & b_1 & 0 & 0 &\dots & 0\\[1ex]
		b_1 & a_2 & b_2 & 0 &\dots & 0\\[1ex]
		0 & b_2 & a_3 & b_3 &\dots & 0\\[1ex]
		0 & 0 & b_3 &\ddots &\ddots & \vdots\\[1ex]
		\vdots & & & \ddots & \ddots & b_{N}\\[1ex]
		0 & 0 & 0 & \dots &b_{N} & ~a_{N+1}~\\
	\end{pmatrix}\,,
	\ee
	\ie it is symmetric tridiagonal. Using the Hamiltonian 
	\begin{equation} \label{eq:Ham_1_Toda}
		H_1(L) = -\frac{1}{2} \,{\rm Tr}\, L^2\,, 
	\end{equation}
	we then find 
	\be
	R_+\nabla H_1(L)=P_+(-L)=\begin{pmatrix}
		0 & b_1 & 0 & 0 &\dots & 0\\[1ex]
		~-b_1~ & 0 & b_2 & 0 &\dots & 0\\[1ex]
		0 & -b_2 & 0 & b_3 &\dots & 0\\[1ex]
		0 & 0 & -b_3 &\ddots &\ddots & \vdots\\[1ex]
		\vdots & & & \ddots & \ddots & ~b_{N}~\\[1ex]
		0 & 0 & 0 & \dots &-b_{N} & 0\\
	\end{pmatrix}\,.
	\ee
	A direct substitution in \eqref{EL_Lax} with $k=1$, \ie
	$$    \partial_{t_1}L=[R_\pm \nabla H_1(L),L]$$
	reproduces the open finite Toda lattice equations in Flaschka's coordinates $a_n$, $b_n$
	\be \label{eq:Flaschka_coord}
	\begin{cases}
		\partial_{t_1}a_1=2b_1^2\,,\qquad \partial_{t_1}a_{N+1}=-2b_{N}^2\,,\\[1.5ex]		\partial_{t_1}a_j=2(b_{j}^2-b_{j-1}^2)\,,\qquad j=2,\dots,N\,, \\[1.5ex]
            \partial_{t_1}b_j=b_j(a_{j+1}-a_j)\,,\qquad j=1,\dots,N\,.
	\end{cases}
	\ee
	The next flow generated by the Hamiltonian 
	\begin{equation}\label{eq:Ham_2_Toda}
		H_2(L) = -\frac{1}{3} \,{\rm Tr}\, L^3\,, 
	\end{equation}
	with gradient $\nabla H_2(L)=-L^2$ yields $R_+\nabla H_2(L)=P_+(-L^2)$ as
	\begin{equation*}
		\begin{pmatrix}
			0 & b_1(a_1+a_2) & b_1\,b_2 & 0 &\dots & 0\\[1ex]
			-b_1(a_1+a_2) & 0 & b_2(a_2+a_3) & b_2\,b_3 &\dots & 0\\[1ex]
			-b_1\,b_2 & -b_2(a_2+a_3) & 0 & b_3(a_3+a_4) &\dots & 0\\[1ex]
			0 & -b_2\,b_3 & -b_3(a_3+a_4) &\ddots &\ddots & \vdots\\[1ex]
			\vdots & & & \ddots & \ddots & b_{N}(a_{N}+a_{N+1})\\[1ex]
			0 & 0 & 0 & \dots &-b_{N}(a_{N}+a_{N+1}) & 0\\
		\end{pmatrix}\,.
	\end{equation*}
	The corresponding equations from \eqref{EL_Lax} with $k=2$ read
	\be \label{eq:second_flow}
	\begin{cases}
\partial_{t_2}a_1=2\,b_1^2(a_1+a_2)\,,\qquad \partial_{t_2}a_{N+1}=-2\,b_{N}^2(a_{N}+a_{N+1})\,,\\[1.8ex]
		\partial_{t_2}a_j=2b_j^2(a_j+a_{j+1})-2\,b_{j-1}^2(a_{j-1}+a_j)\,,\qquad j=2,\dots,N, \\[1.8ex] 
		\partial_{t_2}b_1= b_1(a_2^2-a_1^2+b_2^2\,)  \,,\qquad \partial_{t_2}b_{N}=b_{N}(a_{N+1}^2-a_{N}^2-b_{N-1}^2)\,,\\[1.8ex]
		\partial_{t_2}b_j=b_j(a_{j+1}^2-a_j^2+b_{j+1}^2-b_{j-1}^2)\,,\qquad j=2,\dots,N.
	\end{cases}
	\ee
	\paragraph{Lagrangian description:} We need to choose a convenient parametrisation of $\varphi$ since this is the essential ingredient in the Lagrangians $\Lag_k$. We choose
	\begin{equation}
		\varphi=U\,Y \,, 
	\end{equation}
	where $Y={\rm diag}(y_1,\dots,y_{N+1})$ is the diagonal matrix of diagonal elements of $\varphi$ (\ie $y_i=\varphi_{ii}$) and $U=\varphi\,Y^{-1}$ is the upper triangular matrix with $1$ on the diagonal and arbitrary elements $u_{ij}$, $1\! \le \! i\!<j\!\le \!N$. Since $\varphi$ has non-zero determinant, $y_i\neq 0$, $i=1,\dots,N+1$, and with this parametrisation, we find $L$ as in \eqref{form_L_Flaschka} with
	\begin{equation}\label{ab_uy_vars}
		\begin{cases} 
		a_1=\dfrac{y_{2}}{y_1}\,u_{12}\,,\qquad    
		a_{N+1}=-\dfrac{y_{N+1}}{y_{N}}\,u_{N,N+1}\,, \\[2ex]
            a_i=\dfrac{y_{i+1}}{y_i}\,u_{i,i+1}-\dfrac{y_{i}}{y_{i-1}}\,u_{i-1,i}   \,,\qquad i=2,\dots,N\,,\\[2ex]  
            b_i=\dfrac{y_{i+1}}{y_i}\,,\qquad i=1,\dots,N    \,.
            \end{cases} 
	\end{equation}
	Note that $\displaystyle \sum_{j=1}^{N+1}a_j=0$, so we have $2N$ independent variables on the coadjoint orbit ${\cal O}_\Lambda$. We compute the kinetic part of $\Lag_k$ defined in \eqref{kin_part} as
	\begin{equation}
 \begin{split} 
		K_k&= \langle-R^*_-(\varphi \,\Lambda\,\varphi^{-1})\,,\,\partial_{t_k}\varphi \cdot_R \varphi^{-1}\rangle  = -\langle\varphi \,\Lambda\,\varphi^{-1}\,,\,R_-(\partial_{t_k}\varphi \cdot_R \varphi^{-1})\rangle\\[1ex] 
		&=-\langle\varphi \,\Lambda\,\varphi^{-1}\,,\,\partial_{t_k}\varphi \, \varphi^{-1}\rangle=-{\rm Tr}\left(\Lambda\,\varphi^{-1}\,\partial_{t_k}\varphi \right)\,,
  \end{split} 
	\end{equation}
	where in the third step we have used the morphism property of $R_-$
	\begin{equation}
      R_-(\partial_{t_k}\varphi \cdot_R \varphi^{-1})=\partial_{t_k}\varphi_- \, \varphi_-^{-1}
	\end{equation} 
and $\varphi_-=\varphi\in G_-$.
	It remains to express it in terms of our chosen coordinates to get
	\begin{eqnarray}
		K_k =-{\rm Tr}\left(\Lambda\, Y^{-1}\,U^{-1}\,\partial_{t_k}(UY) \right)=-{\rm Tr}\left(Y\,\Lambda\, Y^{-1}\,U^{-1}\,\partial_{t_k}U \right)= - \sum_{j=1}^{N}\frac{y_{j+1}}{y_j}\,\partial_{t_k}u_{j,j+1}\,.
	\end{eqnarray}
	From these results, it becomes apparent that the convenient coordinates are $b_i$ as given in \eqref{ab_uy_vars} and $u_{i}\equiv u_{i,i+1}$, $i=1,\dots,N$. 
	The first two Lagrangians involve the Hamiltonians \eqref{eq:Ham_1_Toda} and \eqref{eq:Ham_2_Toda} respectively, and can now be expressed in the $u_i,b_i$ coordinates as follows
	\begin{align*}
		\Lag_1= K_1 - H_1 &= -\sum_{j=1}^{N}b_j\,\partial_{t_1}u_{j}+ \frac{1}{2}\sum_{j=2}^{N}(b_j\,u_j-b_{j-1}\,u_{j-1})^2+\sum_{j=1}^{N}b_j^2+ \frac{1}{2}b_1^2\,u_1^2+ \frac{1}{2}b_{N}^2\,u_{N}^2\,, \\
		\Lag_2=K_2 - H_2 &= -\sum_{j=1}^{N}b_j\,\partial_{t_2}u_{j}+\frac{1}{3}\sum_{j=2}^{N}(b_j\,u_j-b_{j-1}\,u_{j-1})^3+\frac{1}{3}(b_1\,u_1)^3+\frac{1}{3}(b_{N}\,u_{N})^3\nonumber\\
		& +\sum_{j=2}^{N-1}b_j^2(b_{j+1}\,u_{j+1}-b_{j-1}\,u_{j-1})+b_1^2(b_2\,u_2)-b_{N}^2(b_{N-1}\,u_{N-1}) \,. 
	\end{align*}
	The variation of $\Lag_1$ reads
	\begin{align*}
		\delta\Lag_1&= -\sum_{j=1}^{N}  \partial_{t_1}u_{j}\,\delta b_j+\sum_{j=1}^{N}  \partial_{t_1}b_j\,\delta u_{j}-  \partial_{t_1} \sum_{j=1}^{N}b_j\,\delta u_{j} \\
		&~~~+ \sum_{j=2}^{N}(b_ju_j-b_{j-1}u_{j-1})(u_j\,\delta b_j+b_j\,\delta u_j)- \sum_{j=1}^{N-1}(b_{j+1}u_{j+1}-b_{j}u_{j})(u_j\,\delta b_j+b_j\,\delta u_j)\\
		&~~~+2\sum_{j=1}^{N}b_j\,\delta b_j+ b_1u_1^2\,\delta b_1+b_1^2u_1\,\delta u_1+ b_{N}u_{N}^2\,\delta b_{N} + b_{N}^2u_{N}\,\delta u_{N} \,, 
	\end{align*}
	and gives the following Euler-Lagrange equations
	\begin{equation}
		\begin{cases}
			\partial_{t_1}u_{1}=u_1^2\,b_1-u_1\,(b_{2}\,u_{2}-b_{1}\,u_{1})+2\,b_1\,,\\[1.5ex]
			\partial_{t_1}u_{N}=u_N\,(b_N\,u_N-b_{N-1}\,u_{N-1})+u_N^2\,b_N+2\,b_N\,,\\[1.5ex]
			\partial_{t_1}u_{j}=u_j(b_j\,u_j-b_{j-1}\,u_{j-1})-u_j(b_{j+1}\,u_{j+1}-b_{j}\,u_{j})+2\,b_j\,,\\[1.5ex]
			\partial_{t_1}b_1=b_1(b_2u_2-b_{1}u_{1})-b_1^2\,u_1       \,,~~\partial_{t_1}b_{N}=-b_{N}(b_{N}u_{N}-b_{N-1}u_{N-1})-b_{N}^2\,u_{N}\,,\\[1.5ex]
			\partial_{t_1}b_j=b_j(b_{j+1}u_{j+1}-b_{j}u_{j})-b_j(b_ju_j-b_{j-1}u_{j-1})\,,    
		\end{cases}
	\end{equation}
 for $j=2,\dots,N-1$. It is easy to see that these equations give exactly \eqref{eq:Flaschka_coord} using the identification (see \eqref{ab_uy_vars})
	\begin{equation}	\label{identification}
 \begin{cases}
     a_1=b_1u_1\,,\qquad a_{N+1}=-b_{N}u_{N}\,,\\[1ex] 
     a_j=b_ju_j-b_{j-1}u_{j-1}\,,\qquad j=2,\dots,N\,.
 \end{cases}
	\end{equation}
	This provides a very explicit check that our Lagrangians produce the corresponding Lax equations, in coordinates naturally dictated by the coadjoint orbit construction of the kinetic term, here $u_j$, $b_j$. As recalled in Section \ref{closure_KK}, the kinetic part of a Lagrangian provides the (pullback of the) symplectic form of the model via the Cartan form $\theta_R$. Here, we have (see the total derivative term in $\delta \Lag_1$)
	\begin{equation}
		\label{omega_b_u}
		\theta_R=\sum_{j=1}^{N}b_j\,\delta u_{j}~\Rightarrow ~\Omega_R=\sum_{j=1}^{N}\delta b_j\wedge\delta u_{j}\,.
	\end{equation}
	This shows that the coordinates $u_j$, $b_j$ are canonical. 
	In the present case, choosing $b_j$, $a_j$ for $j=1,\dots,N$, as the coordinates on the coadjoint orbit ${\cal O}_\Lambda$, we can also express the Kostant-Kirillov form explicitly using the formula
	$$u_j=\frac{1}{b_j}\sum_{\ell=1}^j a_{\ell}$$ to get
	\begin{equation}
		\omega_R=\sum_{j=1}^{N}\frac{1}{b_j}\sum_{\ell=1}^j\delta b_j\wedge\delta a_{\ell}\,.
	\end{equation}
	It is instructive to see how the usual Hamiltonian formulation of the open Toda chain in canonical coordinates $q_i,p_i$ is derived from our Lagrangian formulation. From the symplectic form \eqref{omega_b_u}, we deduce the following (canonical) Poisson brackets\footnote{Here we drop the subscript $R$ when referring to the Poisson bracket $\{\,~,~\}_R$ since there will be no confusion with another Poisson bracket.}
	\begin{equation}
		\label{PB_b_u}
		\{b_j\,,\,u_k\}=\delta_{jk}\,,\qquad \{b_j\,,\,b_k\}=0=\{u_j\,,\,u_k\}\,,\qquad j,k=1,\dots,N\,.
	\end{equation}
	The Legendre transformation
	$$\frac{\partial\Lag_1}{\partial (\partial_{t_1} u_j)}=b_j$$
	reproduces, as it should, the Hamiltonian 
	\begin{equation*}
		\sum_{j=1}^{N}\frac{\partial\Lag_1}{\partial (\partial_{t_1} u_j)}\,\partial_{t_1}u_{j}-\Lag_1=   -\,\frac{1}{2}\sum_{j=2}^{N}(b_j\,u_j-b_{j-1}\,u_{j-1})^2-\sum_{j=1}^{N}b_j^2- \frac{1}{2}b_1^2\,u_1^2- \frac{1}{2}b_{N}^2\,u_{N}^2=H_1(L) \,. 
	\end{equation*}
The matrix $L$ for $\mathfrak{sl}(N+1)$ in canonical coordinates $(q_i,\, p_i)$ is given by 
\begin{equation}\label{form_L_pq}
    L = \begin{pmatrix}
    p_1 & \text{e}^{q_1-q_2} & \hspace{2ex} 0 \hspace{2ex} & 0 & 0 & \dots \\[2ex]
    \text{e}^{q_1-q_2} & p_2 & \text{e}^{q_2-q_3} & 0 & 0 & \dots \\[2ex] 
    0 & \text{e}^{q_2-q_3} & p_3 & \text{e}^{q_3-q_4} & 0 & \dots\\[2ex] 
    0 & 0 & \ddots & \ddots & \ddots & & \\[2ex]
    \vdots & & & \text{e}^{q_{N-1}-q_{N}} & p_n &  \text{e}^{q_{N}-q_{N+1}}\!\!\!\! \\[2ex]
    0 & & &  &  \text{e}^{q_{N}-q_{N+1}} & p_{N+1}
    \end{pmatrix} \,,
\end{equation}
and by comparison with \eqref{form_L_Flaschka}, we set the change of variables
	\begin{equation}
         \begin{cases} 
		    \,q_j=\displaystyle\sum_{k=j}^{N}\ln b_k\,,\qquad j=1,\dots,N\,,\\[-.1ex]
		  \,p_j=b_j\,u_j-b_{j-1}\,u_{j-1}\,,\qquad j=2,\dots,N\,,\\[1.5ex]
		  \,p_1=b_1\,u_1\,,\qquad p_{N+1}=-b_{N}\,u_{N}\,.
         \end{cases}
	\end{equation}
From \eqref{PB_b_u}, we deduce by direct calculation
	\begin{equation}
 \begin{split}
		\label{PB_q_p}
		&\{q_j\,,\,p_k\}=\delta_{jk}\,,~~\{q_j\,,\,q_k\}=0=\{p_j\,,\,p_k\}\,,~~j,k=1,\dots,N\,, \\[1ex]
		&\{q_j\,,\,p_{N+1}\}=-1\,,~~\{p_j\,,\,p_{N+1}\}=0\,,~~j=1,\dots,N\,.
  \end{split} 
	\end{equation}
Note that $p_{N+1}$ is redundant for the description of the dynamics since we only need the map $(u_j,b_j)\mapsto (q_j,p_j)$ for $j=1,\dots,N$. This is captured by the fact that the previous relations imply that $C=\displaystyle \sum_{j=1}^{N+1} p_j$ is a Casimir on the $2N$ phase space with coordinates $(q_1,\dots,q_{N},p_1,\dots,p_{N+1})$ and we can work with $C=0$. The coordinate $p_{N+1}$ is still useful to write the Hamiltonian in the compact familiar form as
\begin{equation} \label{eq:Ham_1_Toda_p_q}
		H_1= -\frac{1}{2}\sum_{j=1}^{N+1}p_j^2-\sum_{j=1}^{N-1}\text{e}^{2(q_j-q_{j+1})}-\text{e}^{2q_{N}}\,.
\end{equation}
Hamilton's equations $\partial_{t_1}q_j=\{q_j,H_1\}$, $\partial_{t_1}p_j=\{p_j,H_1\}$ yield
	\begin{equation}\label{eq:Toda_p_q}
		\begin{cases}
			\partial_{t_1}p_1=2\text{e}^{2(q_{1}-q_2)}\,,\qquad \partial_{t_1}p_{N+1}=-2\text{e}^{2q_{N}}\,,\\[1.2ex]
			\partial_{t_1}p_j=2\left(\text{e}^{2(q_j-q_{j+1})}-\text{e}^{2(q_{j-1}-q_j)}\right)\,,\qquad j=2,\dots,N-1\,,\\[1.2ex]
			\partial_{t_1}q_j=p_{N+1}-p_j\,,~~j=1,\dots,N\,.
		\end{cases}
	\end{equation}  
These can be seen to be equivalent to \eqref{eq:Flaschka_coord}, thus completing the Hamiltonian description of the first flow for the open Toda chain, from our Lagrangian formulation. The same analysis can be performed with $\Lag_2$ although the calculations are longer. We simply record here the Euler-Lagrange equations obtained from $\delta \Lag_2$
	\begin{equation}
 \begin{cases}
		\partial_{t_2}u_1= u_1\left((b_{1}\,u_{1})^2-(b_{2}\,u_{2}-b_{1}\,u_{1})^2-b_{2}^2\,\right)  +2\,b_1\,b_{2}\,u_{2}  \,,\\[1.5ex]
		\partial_{t_2}u_{N}=u_N\left((b_{N}\,u_{N}-b_{N-1}\,u_{N-1})^2-(b_{N}\,u_{N})^2+b_{N-1}^2\,\right) -2\,b_{N}\,b_{N-1}\,u_{N-1}\,,\\[1.5ex]
		\partial_{t_2}u_j=u_j\left((b_{j}\,u_{j}-b_{j-1}\,u_{j-1})^2-(b_{j+1}\,u_{j+1}-b_{j}\,u_{j})^2\right)+u_j\,(b_{j-1}^2-b_{j+1}^2) \\[1ex] 
        \hspace{10ex} +2\,b_j\,(b_{j+1}\,u_{j+1}-b_{j-1}\,u_{j-1})\,,\\[1.5ex]
		\partial_{t_2}b_1= b_1\left((b_{2}\,u_{2}-b_{1}\,u_{1})^2-(b_{1}\,u_{1})^2+b_{2}^2\,\right) \,,\\[1.5ex]
		\partial_{t_2}b_{N}= b_N\left((b_{N}\,u_{N})^2-(b_{N}\,u_{N}-b_{N-1}\,u_{N-1})^2-b_{N-1}^2\,\right) \,,\\[1.5ex]
		\partial_{t_2}b_j= b_j\left((b_{j+1}\,u_{j+1}-b_{j}\,u_{j})^2-(b_{j}\,u_{j}-b_{j-1}\,u_{j-1})^2\right)-b_j\,(b_{j-1}^2-b_{j+1}^2)  \,,
        \end{cases} 
	\end{equation}
 for $j=2,\dots,N-1$. We leave it to the reader to check that these correctly reproduce \eqref{eq:second_flow} again using \eqref{identification}. 
		To conclude this example we establish the closure relation for the first two flows, \ie
	$$\partial_{t_2}\Lag_1-\partial_{t_1}\Lag_2=0~~\text{on shell}\,.$$
	We know from our general results that this must hold, so this is simply an explicit check. We know that the kinetic and potential contributions give zero separately, so we split the calculations accordingly. For the potential terms, it is more expedient to use the $a_j,b_j$ coordinates\footnote{Note that for conciseness, we treated the equations for $j=1$ and $j=N$ on the same level as for $j=2,\dots,N-1$ by formally introducing $b_0=0$ and $b_{N+1}$.} and equations  \eqref{eq:Flaschka_coord} and \eqref{eq:second_flow}
	\begin{align*}
		\partial_{t_2}H_1-\partial_{t_1}H_2 &= \partial_{t_1}\left(\sum_{j=1}^{N+1}\frac{a_j^3}{3}+\sum_{j=1}^{N}b_j^2(a_j+a_{j+1})  \right)-\partial_{t_2}\left(\sum_{j=1}^{N+1}\frac{a_j^2}{2}+\sum_{j=1}^{N}b_j^2  \right)\\
		&= \sum_{j=1}^{N+1} 2 a_j^2(b_{j-1}^2-b_j^2) +
  \sum_{j=1}^{N}2b_j^2(a_j^2-a_{j+1}^2+b_{j-1}^2-b_{j+1}^2)\\
		&\hspace{2ex} -\sum_{j=1}^{N+1} 2 a_j(b_{j-1}^2(a_{j-1}+a_j)-b_j^2(a_j+a_{j+1})) -\sum_{j=1}^{N}2b_j^2(a_j^2-a_{j+1}^2+b_{j-1}^2-b_{j+1}^2)\\
		&=\sum_{j=1}^{N+1} 2 (a_ja_{j+1}b_j^2-a_ja_{j-1}b_{j-1}^2 )\\
		&=0\,,
	\end{align*}
	where in the last step we recognise a telescopic sum. For the kinetic terms, we also use the $a_j$, $b_j$ coordinates wherever possible to expedite the calculations
	\begin{align*}
		    & \partial_{t_1}K_2 - \partial_{t_2}K_1 = \sum_{j=1}^N(\partial_{t_1} (b_j\partial_{t_2}u_j)-\partial_{t_2}(b_j\partial_{t_1}u_j )) \\[-.5ex]
		&\hspace{5ex} =\sum_{j=1}^N(\partial_{t_1}((a_{j+1}^2-a_j^2+b_{j+1}-b_{j-1}^2)u_jb_j-2b_j^2(a_j+a_{j+1}) )-\partial_{t_2}((a_{j+1}-a_j)u_jb_j-2b_j^2 ))\\
		&\hspace{5ex}=\sum_{j=1}^N((\partial_{t_1}(a_{j+1}^2-a_j^2+b_{j+1}-b_{j-1}^2)- \partial_{t_2}(a_{j+1}-a_j)) u_jb_j  - 2b_j^2(b_{j+1}^2-b_{j-1}^2))\\
		&\hspace{5ex}=0 \,, 
	\end{align*}
	where in the last step the first term gives zero for each $j$ upon using the equations of motion and the remaining terms form a telescopic sum adding up to zero. 

 A Lagrangian multiform for the Toda chain was first constructed in \cite{PS} using variational symmetries of a given starting Lagrangian, which would be $\Lag_1$ in our context, to construct higher Lagrangian coefficients which constitute a multiform when assembled together. The infinite Toda chain was studied more recently in \cite{SV2} to illustrate the newly introduced theory of Lagrangian multiforms over semi-discrete multi-time. In \cite{PS}, the analogue of our $\Lag_2$ and $\Lag_3$ were constructed. The Noether integrals $J_1$ and $J_2$ (equations (10.11) and (10.12) in \cite{PS}) which constitute the potential part of their Lagrangians are nothing but $H_2(L)$ and $H_3(L)$ with $L$ parametrised as in \eqref{form_L_pq}, up to an irrelevant change of convention $\text{e}^{q_i-q_{i+1}}\to \text{e}^{q_{i+1}-q_i}$ and setting $q_i=x_i$ and $p_i=\dot{x}_i$. The kinetic part of the higher Lagrangians in \cite{PS} involves the so-called alien derivatives which are symptomatic of constructing a multiform from a starting Lagrangian and building compatible higher Lagrangian coefficients. Our construction prevents the problem of alien derivatives altogether, putting all the Lagrangian coefficients on equal footing. This was also achieved previously in the context of field theories in \cite{CS2,CSV}.

\section[Open Toda chain with a skew-symmetric \texorpdfstring{$\boldsymbol{r}$-matrix}{r-matrix}]{Open Toda chain with a skew-symmetric $\boldsymbol{r}$-matrix }\label{Toda_pq}
We now present the same model for the same algebra $\g=\sl(N+1)$ but endowed with a different Lie dialgebra structure. This is based on the Cartan decomposition of $\g$ and leads to a skew-symmetric $r$-matrix. One attractive feature of this setup, that we only illustrate for $\sl(N+1)$, is that it allows for a generalisation to any finite semi-simple Lie algebra, see \cite[Chapter 4]{BBT}. 
\paragraph{Algebraic setup:}	Consider the decomposition 
	\begin{equation}  
 \label{decomp_g}
		\mathfrak{g} =  \mathfrak{n}_+ \oplus \mathfrak{h} \oplus \mathfrak{n}_- \,, 
	\end{equation}
	where $\mathfrak{h}$ is the Cartan subalgebra of diagonal (traceless) matrices and $\mathfrak{n}_\pm$ the nilpotent subalgebra of strictly upper/lower triangular matrices. 
 Let $P_\pm$, $P_0$ be the projectors onto $\mathfrak{n}_\pm$ and $\mathfrak{h}$ respectively, relative to the decomposition \eqref{decomp_g} and set $R=P_+-P_-$. It can be verified that $R$ satisfies the mCYBE. Here $R_{\pm}=\pm(P_\pm+P_0/2)$ and 	
 \begin{equation}
		\mathfrak{g}_{\pm} = \text{Im}(R_{\pm}) = \mathfrak{b}_{\pm} = \mathfrak{h} \oplus \mathfrak{n}_{\pm}\,. 
	\end{equation}
We have the following action of $R_\pm$ on the elements $y \in \mathfrak{h}$ and $w_{\pm} \in \mathfrak{n}_{\pm}$,
	\begin{equation} \label{eq:R_action_BBT}
		R_{\pm} (y) = \pm \frac{1}{2}\, y \,,\qquad 
		R_{\pm} (w_{\pm}) = \pm \, w_{\pm} \,,\qquad 
		R_{\pm} (w_{\mp}) = 0 \,.
	\end{equation}
 Taking the same bilinear form as in \eqref{bilinear_form}, \ie $\langle X, Y\rangle=\Tr(XY)$, we see that 
 \begin{eqnarray}
 \label{adjoints}
     P_\pm^*=P_\mp\,,~~P_0^*=P_0~~\text{so that}~~ R^*=-R\,.
 \end{eqnarray}
 Thus, we have a skew-symmetric $r$-matrix here. 
	For the related Lie groups, we have the following factorisations close to the identity, 
	\begin{equation} \label{eq:phi_BBT}
		\varphi =  \varphi_{+}\,\varphi_-^{-1} \,, \qquad \varphi_{\pm} =  W_{\pm}\,Y^{\pm 1} \,, \qquad Y\in \text{exp}(\mathfrak{h}) \,,\,\, W_{\pm}\in\text{exp}(\mathfrak{n}_{\pm}) \,. 
	\end{equation}
\paragraph{Lax matrix and Lax equations for the first two flows:}	For $\Lambda \in \mathfrak{g}^*\simeq \g$, the expression of $L$ as a coadjoint orbit of $\Lambda$ is given by
	\begin{equation}
		\begin{split} \label{eq:L_mf_toda_complete}
			L &= \text{Ad}^{R*}_{\varphi} \cdot \Lambda = R^{^*}_+( W_+\,Y\,\Lambda\,Y^{-1}\,W_+^{-1}) - R^{^*}_-(W_-\,Y^{-1} \,\Lambda\,Y\,W_-^{-1}) \,.
		\end{split}
	\end{equation}
	We choose $\Lambda$ as in \eqref{eq:Lambda}, emphasising that in this case it is an element of the full $\mathfrak{g}^* \simeq \mathfrak{g}$, and $Y \in \text{exp}(\mathfrak{h})$, $W_{\pm}\in \text{exp}(\mathfrak{n}_{\pm})$ given by
	\begin{align} \label{eq:Y_BBT}
		&\hspace{20ex}Y = \text{diag} \left( \eta_1 \,, \eta_2 \dots, \eta_{N+1} \right) \,, \qquad \det Y = 1\,, \\[2ex]
	 \label{eq:w_pm_BBT}
		&W_- \!=\! {\small \begin{pmatrix} 
			1 & 0  & 0 &\dots & 0\\[1.5ex]
			\omega^-_{2,1} & 1  & 0 &\dots & 0\\[1.5ex]
			\omega^-_{3,1} & \omega^-_{3,2} & 1 &\dots & 0\\[1.5ex]
			\vdots & \ddots & \ddots & \ddots & ~~0~~\\[1.5ex]
			~\omega^-_{N,1}~ & \omega^-_{N,2}  & \dots &\omega^-_{N,N-1} & 1\\
		\end{pmatrix}} , ~~ W_+ \!=\! {\small \begin{pmatrix}
			1 & \omega^+_{1, 2} & \omega^+_{1, 3}  &\dots & \omega^+_{1, N}\\[1.5ex]
			0 & 1 & \omega^+_{2, 3}  &\dots & \omega^+_{2 ,N}\\[1.5ex]
			0 & 0 & 1 &\ddots & \vdots\\[1.5ex]
			\vdots & &  & \ddots & ~\omega^+_{N-1, N}~\\[1.5ex]
			~~0~~ & 0  & \dots & 0 & 1\\
		\end{pmatrix}} .
	\end{align}
	From \eqref{adjoints}, we deduce that $R^*_{\pm}=\pm(P_\mp+P_0/2)$ so that
	\begin{equation}
			R_{\pm}^* (y) = \pm \frac{1}{2}\, y \,,\qquad 
			R_{\pm}^* (w_{\pm}) = 0 \,,\qquad  
			R_{\pm}^* (w_{\mp}) =  \pm w_{\mp} \,,
	\end{equation}
 for $y \in \mathfrak{h}\,,~w_{\pm} \in \mathfrak{n}_{\pm}$.  Let us introduce the variables $(w_i,z_i)$, defined as  
	\begin{equation} \label{eq:variables_w_z}
		w_i = \frac{\omega^+_{i,i+1} - \omega^-_{i+1,i}}{2}\,, \qquad z_i=2\,\frac{\eta_{i+1}}{\eta_i} \,,
	\end{equation}
from which we determine the Flaschka coordinates as 
	\begin{equation} \label{eq:BBT_Flash_coord}
 \begin{cases}
     	a_i = \dfrac{w_i\,z_i-w_{i-1}\, z_{i-1}}{2} \,, \qquad i = 2, \dots, N-1\,, \\[2ex]
      a_1 = \dfrac{w_1\,z_1}{2} \,, \qquad a_{N+1} = -\dfrac{w_N\,z_N}{2} \,,\\[2ex]
      b_i = \dfrac{z_i}{2} \,,\qquad  i=1, \dots, N\,.  
 \end{cases}
\end{equation}
The evaluation of \eqref{eq:L_mf_toda_complete} in those coordinates reproduces the tridiagonal form as in \eqref{form_L_Flaschka}. One can then check that the equations for the first two flows \eqref{eq:Flaschka_coord} and \eqref{eq:second_flow} in the previous section derive from the Lax equation 
\begin{equation*}
    \partial_{t_k} L = \left[\, R_{+}(\nabla H_k(L)),L \,\right] \,, \qquad k = 1,2 \,,
\end{equation*}
where the Hamiltonians are taken as 
\begin{equation}
    H_1(L) = \Tr\,(L^2)\,, \qquad  H_2(L) = \frac{2}{3} \Tr\,(L^3)\,,
\end{equation}
and we recall that $R_+=P_+ +P_0/2$ here.

\paragraph{Lagrangian description:} The Lagrangian multiform takes the form
	\begin{equation} \label{eq:Lagrangian_Toda_a_la_BBT}
		\Lag = \sum_k \Lag_{k} \, dt_k = \sum_k \left( {K}_k(L)  - {H}_k(L) \right) dt_k \,, 
	\end{equation}
for $L \in \mathcal{O}_{\Lambda} \,, \varphi \in G_R\,$, with the kinetic and the potential terms given by
	\begin{equation} \label{eq:kinetic_potential_terms_BBT}
		{K}_k(L) = \text{Tr} \,( L \,  \partial_{t_k} \varphi \cdot_R \varphi^{-1} )\,, \qquad  {H}_k(L) = \frac{2}{k+1} \,\text{Tr}\,(L^{k+1}\,)\,, 
	\end{equation} 
 respectively. As in the previous section, the kinetic term will allow us to recognise natural canonical variables of the system in this description. Recalling \eqref{eq:phi_BBT}, \eqref{eq:Y_BBT}, \eqref{eq:w_pm_BBT} and \eqref{eq:variables_w_z}, we find
	\begin{equation}
 \label{eq:BBT_kinetic_term}
 \begin{split} 
			{K}_k(L)  &=  \,\text{Tr}\,(\, \Lambda \, \varphi^{-1} \cdot_R \partial_{t_k} \varphi  \,) \\[1.2ex]
			&= \,  \,\text{Tr}\,(\, \Lambda \, \varphi^{-1}_+ \cdot \partial_{t_k} \varphi_+  \, ) -  \,\text{Tr}\,(\, \Lambda \, \varphi^{-1}_- \cdot \partial_{t_k} \varphi_-  \, ) \\
   &=\sum_{i=1}^N \dfrac{\eta_{i+1}}{\eta_i} \, \partial_{t_k} \omega_{i,i+1}^+ -\sum_{i=1}^N \dfrac{\eta_{i+1}}{\eta_i} \, \partial_{t_k} \omega_{i+1,i}^-  \\
   &=\sum_{i=1}^N z_i \, \partial_{t_k} w_i \,.
   \end{split} 
	\end{equation} 
	The $k$-th Lagrangian coefficient expressed in terms of the coordinates $(w_i,z_i)$ reads
	\begin{equation}
		\Lag_k = {K}_k - {H}_k = \sum_{i=1}^N     z_i\, \partial_{t_k}w_i - {H}_k  \,,
	\end{equation}
	with $(w_i,z_i)$ being canonical coordinates, and for $k=1,2$,
	\begin{equation*}
		\begin{split}
			{H}_1(L) &=  \,\text{Tr}\,(L^2) = \sum_{i=1}^{N} \frac{1}{2}\left(z_i^2+w_i^2\,z_i^2 \right)- \sum_{i=1}^{N-1} \frac{1}{2}\,w_i\,z_i\,w_{i+1}\,z_{i+1}    \,, \\
			{H}_2(L) &= \dfrac{2}{3} \,\text{Tr}\,(L^3) = \sum_{i=1}^{N} \frac{1}{4}\left(z_i^2\,w_{i+1}\,z_{i+1}-z_{i+1}^2\,w_i\,z_i + w_i^2\,z_i^2 \,w_{i+1}\,z_{i+1} -w_i\,z_i\,w_{i+1}^2\,z_{i+1}^2 \right)  \,.
		\end{split}
	\end{equation*}
To obtain the latter expressions of the Hamiltonians, it suffices to use the following expression for $L$ in the $w_i,z_i$ coordinates
\begin{equation} 
		L =  \frac{1}{2}\begin{pmatrix}
			~w_1\,z_1 & \,z_1 & 0 & 0 &\dots & 0\\[2ex]
			\,z_1 & w_2\,z_2-w_1\,z_1 & \,z_2 & 0 &\dots & 0\\[2ex]
			0 & \,z_2 & w_3\,z_3-w_2\,z_2 & \,z_3 &\dots & 0\\[2ex]
			0 & 0 & \,z_3 &\ddots &\ddots & \vdots\\[2ex]
			\vdots & & & \ddots & \ddots & \,z_N\\[2ex]
			0 & 0 & 0 & \dots & \,z_N & -w_N\,z_N ~
		\end{pmatrix} \,.
  \label{eq:L_BBT_R_skew}
	\end{equation}
 Note that we can just as easily determine the higher Hamiltonians and hence the higher Lagrangian coefficients $\Lag_k$, although the expressions become long. 
The variation $\delta \Lag_1$ yields the following Euler-Lagrange equations
	\begin{equation}
		\begin{cases} \label{eq:BBT_t1_w_z_R}
  \partial_{t_1} w_1 = z_1 - \dfrac{w_1}{2} \,  \left((w_{2}\, z_{2} - w_{1}\, z_{1}) -  w_{1}\, z_{1}  \right)  \,, \\[2ex] 
  \partial_{t_1} w_N =   z_N - \dfrac{w_N}{2}\,  \left( - w_{N}\, z_{N} - ( w_{N}\, z_{N} - w_{N-1}\, z_{N-1}   ) \right)  \,, \\[2ex]
\partial_{t_1} w_i =  z_i - \dfrac{w_i}{2}\,  \left( (w_{i+1}\, z_{i+1} -w_{i}\, z_{i}  ) -(w_{i}\, z_{i} - w_{i-1}\, z_{i-1} ) \right)\,,  \\[2ex]
			\partial_{t_1} z_1 = \dfrac{z_1}{2} \,  \left((w_{2}\, z_{2} - w_{1}\, z_{1}) - w_{1}\, z_{1} \right)\,, \\[2ex] 
   \partial_{t_1} z_N =  \dfrac{z_N}{2} \,  \left( -w_{N}\, z_{N} +(w_{N}\, z_{N} - w_{N-1}\, z_{N-1}) \right)\,,\\[2ex]
        \partial_{t_1} z_i =  \dfrac{z_i}{2} \,  \left((w_{i+1}\, z_{i+1}-  w_{i}\, z_{i}) - ( w_{i}\, z_{i} - w_{i-1}\, z_{i-1} )  \right)\,,
		\end{cases}
	\end{equation}
for $i=2,\,\dots,\,N-1$, while the variation of $\Lag_2$ gives the Euler-Lagrange equations for the second flow
\begin{equation} 
    \begin{cases}
        \partial_{t_2} z_1 = \dfrac{z_1}{4} \big( (w_{2}\,z_{2}-w_1\,z_1)^2 -(w_{1}\,z_{1})^2 +  z_{2}^2  \big) \,,  \\[2ex]
        \partial_{t_2} z_N = \dfrac{z_N}{4} \left( (-w_N\,z_N)^2 - (w_N\,z_N - w_{N-1}\,z_{N-1})^2 -z_{N-1}^2 \right) \,, \\[2ex]
        \partial_{t_2} z_i = \dfrac{z_i}{4}\big( (w_{i+1}\,z_{i+1}-w_i\,z_i)^2 -(w_{i}\,z_{i}-w_{i-1}\,z_{i-1})^2 +  z_{i+1}^2- z_{i-1}^2  \big)\,,\\[2ex] 
        \partial_{t_2} w_1 =  \dfrac{z_1}{2}\big(w_{2}\,z_{2}\big) -\dfrac{w_1}{4} \Big( (w_{2}\,z_{2}-w_1\,z_1)^2 -(w_{1}\,z_{1})^2 +  z_{2}^2  \Big) \,,\\[2ex]
        \partial_{t_2} w_N =  \dfrac{z_N}{2}\big(w_{N-1}\,z_{N-1}\big)  -\dfrac{w_N}{4} \Big( (w_N\,z_N)^2 -(w_{N}\,z_{N}-w_{N-1}\,z_{N-1})^2- z_{N-1}^2  \Big) \,,\\[2ex]
        \partial_{t_2} w_i = \dfrac{z_i}{2}\big((w_{i+1}\,z_{i+1}-w_i\,z_i)-(w_i\,z_i-w_{i-1}\,z_{i-1})\big)  -\dfrac{w_i}{4} \Big( (w_{i+1}\,z_{i+1}-w_i\,z_i)^2 \\[2ex] 
        \hspace{15ex} -(w_{i}\,z_{i}-w_{i-1}\,z_{i-1})^2 +  z_{i+1}^2- z_{i-1}^2  \Big) \,,
    \end{cases}
\end{equation}
with $i=1,\dots,N-1$. One can check that these reproduce the more familiar equations \eqref{eq:Flaschka_coord}-\eqref{eq:second_flow} in Flaschka coordinates, using \eqref{eq:BBT_Flash_coord}.
As in the previous section, we can relate our results with the Hamiltonian formulation of the Toda chain in traditional canonical coordinates $(q_i,p_i)$. With 
\begin{equation}
    \theta_R = \sum_{i=1}^N z_i\,\delta w_i ~~ \implies ~~ \{z_i,w_j\} = \delta_{ij} \,, \qquad \{w_i,w_j\} = 0 = \{z_i,z_j\} \,,
\end{equation}
we see that it suffices to set
\begin{equation}
\begin{cases}
    q_i = \displaystyle \sum_{\ell=i}^N \ln \dfrac{z_i}{2}\,, \qquad i = 1, \dots, N \\[4ex]
    p_i = \dfrac{w_i\,z_i-w_{i-1}\,z_{i-1}}{2}\,, \qquad  i = 2, \dots, N \\[2ex]
    p_1 = \dfrac{w_1\,z_1}{2} \,, \qquad p_{N+1} = -\dfrac{w_N\,z_N}{2} \,.
\end{cases}
\end{equation}

The explicit verification of the closure relation in the first two flows is completely analogous to that given at the end of the previous section.

\section{Rational Gaudin model}\label{Gaudin}

Gaudin models are a general class of integrable systems associated with Lie algebras with a nondegenerate invariant bilinear form. Unlike the case of the open Toda lattice, the Lax matrix of a Gaudin model is a Lie algebra-valued rational function of a variable $\lambda$, the spectral parameter. We will only look at finite Gaudin models here, which describe certain spin chains and mechanical systems. To accommodate this, we need to extend our construction to certain infinite-dimensional Lie algebras.

Before diving into the required algebraic machinery, it is useful to recall the usual presentation of the equations of the model that we are aiming at describing variationally. We do so in the simplest case of a rational Lax matrix with simple poles. Many generalisations are known, including elliptic and non-skew-symmetric cases \cite{Skr}. The Lax matrix of a (rational) Gaudin model associated with a finite Lie algebra $\g$ and a set of points $\zeta_r \in \mathbb{C}$ $(r=1, \ldots, N)$ and the point at infinity is given by the following $\g$-valued rational function
\begin{equation}
\label{Lax_matrix_Gaudin}
    L(\lda)=  \sum_{r=1}^N \frac{X_r}{\lambda - \zeta_r} + X_\infty\,,~~X_1,\dots,X_N,X_\infty\in\g\,.
\end{equation}
The coefficients $H^n_{k, r}$ of $(\lda-\zeta_r)^{-n-1}$, $n\ge 0$, in $\Tr(L(\lda)^{k+1})/{(k+1)}$, $ k\geq 1$, 
are Hamiltonians in involution (with respect to the Sklyanin bracket). Of course, only a finite subset of them are independent and generate nontrivial flows. In the rest of this paper, we will focus on the coefficients corresponding to $n=0$ and drop the extra label by simply writing $H^0_{k, r}=H_{k, r}$. 
The most famous ones are the quadratic Gaudin Hamiltonians which are the coefficients $H_{1, r}$ in
\begin{equation}
    \frac{1}{2}\Tr(L(\lda)^2) = \frac{1}{2}\sum_{r=1}^N\frac{\Tr (X_r^2)}{(\lambda - \zeta_r)^2} + \sum_{r=1}^N\frac{H_{1, r}}{\lambda - \zeta_r} + \frac{1}{2}\Tr(X_\infty^2)\,,
\end{equation}
and read
\begin{eqnarray}
\label{Gaudin_Ham1}
H_{1, r}=\sum_{s\neq r}\frac{\Tr (X_rX_s)}{\zeta_r-\zeta_s}+\Tr(X_rX_\infty)\,,~~r=1,\dots,N\,.
\end{eqnarray}
The functions $H_{k, r}$ give rise to a hierarchy of compatible equations in Lax form
\begin{eqnarray}\label{generating_Gaudin_eqs}
    \partial_{t_k^r}L(\lda)=\left[M_{k, r}(\lda)\,,\,L(\lda)\right]\,.
\end{eqnarray}
For $k=1$, we have
\begin{equation}
\label{first_M}
    M_{1, r} = -\frac{X_r}{\lda - \zeta_r}\,,
\end{equation}
and \eqref{generating_Gaudin_eqs} gives the following equations of motion for the degrees of freedom in $X_1,\dots,X_N,X_\infty$
\begin{align}
\label{Gaudin_1}
   & \partial_{t_1^r}X_s =  \frac{[X_r\,,\,X_s]}{\zeta_r-\zeta_s}\,,\qquad s\neq r\,,\\[1ex]
   \label{Gaudin_2}
   & \partial_{t_1^r}X_r =  -\sum_{s\neq r}\frac{[X_r\,,\,X_s]}{\zeta_r-\zeta_s}- [X_r\,,\,X_\infty] \,,\\[1ex]
\label{Gaudin_3}   & \partial_{t_1^r}X_\infty = 0\,.
\end{align}
We proceed to derive a Lagrangian multiform description of the set of equations \eqref{Gaudin_1}-\eqref{Gaudin_3}, as well as those corresponding to the next higher Hamiltonians with $k=2$. In principle, we could also include all higher Hamiltonians, but the first two levels are enough to illustrate our method. To do so, we need to be able to interpret $L(\lda)$ as living in a coadjoint orbit and use the framework of Lie dialgebras. This is described in \cite[Lecture 3]{STS} which we now review and adapt to our purposes.

\paragraph{Algebraic setup:} Let $Q = \{\zeta_1, \ldots, \zeta_N, \infty\} \subset \mathbb{C}P^1$ be a finite set of points in $\mathbb{C}P^1$ including the point at infinity, and denote by $\mathcal{F}_Q(\g)$ the algebra of $\g$-valued rational functions in the formal variable $\lambda$ with poles in $Q$. Further, define the local parameters
\begin{equation}
  \lambda_r = \lambda - \zeta_r,\quad \zeta_r \neq \infty\,,\qquad  \lambda_\infty = \frac{1}{\lambda}\,,
\end{equation}
and let $S=\{1,\dots,N,\infty\}$. This is to be used as an index set, so $\infty$ is viewed here purely as a label for an index, not as the point at infinity. For each $r\in S$, consider the algebra $\Lg_r$ of formal Laurent series in variable $\lambda_r$ with coefficients in $\g$,
\begin{equation}
  \Lg_r = \g\otimes \mathbb{C}((\lambda_r))\,,
\end{equation}
with Lie bracket
\begin{equation}
    [X\lambda_r^i, Y\lambda_r^j] = [X, Y] \lambda_r^{i+j},\qquad X,Y \in \g\,.
\end{equation}
We have the vector space decomposition into Lie subalgebras
\begin{equation}
    \Lg_r = \Lg_{r+} \oplus \Lg_{r-}\,,
\end{equation}
where
\begin{equation}
    \Lg_{r+}=\g\otimes \mathbb{C}[[\lambda_r]]\,,\quad r\neq \infty\,,\qquad \Lg_{\infty +}=\g\otimes \lda_\infty\mathbb{C}[[\lambda_\infty]]\,,
\end{equation}
and 
\begin{equation}
    \Lg_{r-}=\g\otimes \lambda_r^{-1}\mathbb{C}[\lambda_r^{-1}]\,,\quad r\neq \infty\,,\qquad \Lg_{\infty -}=\g\otimes \mathbb{C}[\lambda_\infty^{-1}]\,.
\end{equation}
In other words, $\Lg_{r+}$ is the algebra of formal Taylor series in $\lda_r$ (without constant term when $r=\infty$) and $\Lg_{r-}$ is the algebra of polynomials in $\lda_r^{-1}$ without constant term (except when $r=\infty$). Associated with this decomposition, we have projectors $P_{r\pm}$ onto $\Lg_{r\pm}$ relative to $\Lg_{r\mp}$.
Let us now consider $\Lg_Q$ defined as the following direct sum of Lie algebras
\begin{equation}
  \Lg_Q = \bigoplus_{r\in S} \Lg_r\,.
\end{equation}
The above decompositions yield the decomposition of $\Lg_Q$ as
\begin{equation}
\label{decomp_Lg_Q} 
\Lg_{Q}=\Lg_{Q+} \oplus \Lg_{Q-}~~\text{with}~~ \Lg_{Q+} = \bigoplus_{r\in S} \Lg_{r+}  ~~ \text{and} ~~ \Lg_{Q-} = \bigoplus_{r\in S} \Lg_{r-}\,,
\end{equation}
and the related projectors $P_\pm$. 
Although useful, as we will see below, the decomposition \eqref{decomp_Lg_Q} is not what we need to interpret \eqref{generating_Gaudin_eqs} within the Lie dialgebra setup. So, let us consider the map 
\begin{equation}\label{eq:embed-rat-loop}
  \iota_\lambda: \mathcal{F}_Q(\g) \to \Lg_Q, \qquad  f \mapsto \left(\iota_{\lambda_1} f, \ldots, \iota_{\lambda_N} f, \iota_{\lambda_\infty} f \right),
\end{equation}
where $\iota_{\lambda_r} f \in \Lg_r$ is the formal Laurent series of $f \in \mathcal{F}_Q(\g)$ at $\zeta_r \in \mathbb{C}P^1$ and $\iota_{\lambda_\infty} f \in \Lg_r$ that of $f \in \mathcal{F}_Q(\g)$ at $\zeta_\infty$. This is an embedding of Lie algebras. In addition, we have the vector space decomposition 
\begin{equation}
\label{decomp2}
  \Lg_Q = \Lg_{Q+} \oplus \iota_\lambda \mathcal{F}_Q(\g).
\end{equation}
 Let us introduce the projectors $\Pi_\pm$ associated with this decomposition. They are different from $P_\pm$ related to \eqref{decomp_Lg_Q}. The following relation is useful in practical calculations (see below when computing gradients or in \eqref{formula_orbit_elem})
 \begin{eqnarray}
    \Pi_-(X)=\iota_\lda \circ \pi_\lda \circ P_-(X)\,,~~X\in\Lg_Q\,,
\end{eqnarray}
where the map $\pi_\lda:\Lg_{Q-}\to   \mathcal{F}_Q(\g)$ given by
\begin{eqnarray}
\pi_\lda\,(Y_1(\lda_1),\dots, Y_N(\lda_N),Y_\infty(\lda_\infty))  =\sum_{r\in S}Y_r(\lda_r)
\end{eqnarray}
puts elements of $\Lg_{Q-}$ and $\mathcal{F}_Q(\g)$ in one-to-one correspondence.
This amounts to decomposing an $f\in \mathcal{F}_Q(\g)$ into the sum of its partial fractions $Y_r(\lda_r)$.
 
We define the $r$-matrix we need as
\begin{equation}
    R=\Pi_+-\Pi_-
\end{equation}
and use it to define on $\Lg_Q$ the structure of a Lie dialgebra to which we will apply the results of that theory. Since we want to work with rational fractions which we have naturally embedded as  $\iota_\lambda \mathcal{F}_Q(\g)$ into $\Lg_Q$, we need to identify the dual space this corresponds to, so that we can identify the coadjoint action and its orbits appropriately. 
The nondegenerate invariant symmetric bilinear form on $\g$, given by $(X, Y) \mapsto \Tr(XY)$, can be used to define a nondegenerate invariant symmetric bilinear form on  $\Lg_Q$ by setting
\begin{equation}\label{eq:bf-lg-q}
  \langle X, Y \rangle = \sum_{r\in S} \mathrm{res}_{\lambda_r = 0} \Tr(X_r(\lambda_r) Y_r(\lambda_r))\,.
\end{equation}
Both $\Lg_{Q+}$ and $\iota_\lambda \mathcal{F}_Q(\g)$ are Lie subalgebras which are (maximally) isotropic with respect to the bilinear form $\langle \, ~~,~~ \rangle$ in \eqref{eq:bf-lg-q}. This tells us that 
\begin{equation}
\label{identification2}
  \Lg_{Q+}^{*} \simeq \iota_\lambda \mathcal{F}_Q(\g)\,,
\end{equation}
so that elements of $\Lg_{Q+}^{*}$ are those we should work with if we want to deal with Lax matrices which are rational fractions of the spectral parameter. Accordingly, coadjoint orbits of $\LG_{Q+}$ in $\Lg_{Q+}^{*}$ are the natural arena for the description of Gaudin Lax matrices. $\LG_{Q+}$
is the group associated with the algebra $\Lg_{Q+}$, with elements of the form
\begin{equation}
    \varphi_+=\left(\varphi_{1+}(\lda_1),\dots,\varphi_{N+}(\lda_N),\varphi_{\infty+}(\lda_\infty) \right)\,.
\end{equation}
Each component $\varphi_{r+}(\lda_r)$ is a Taylor series in the local parameter $\lda_r$ with values in $G$ whose Lie algebra is $\g$,
\begin{equation}
    \varphi_{r+}(\lda_r)=\sum_{n=0}^\infty \phi_r^{(n)}\lda_r^n\,,\quad r\neq \infty\,,\qquad \varphi_{\infty+}(\lda_\infty)=\1+\sum_{n=1}^\infty\phi_\infty^{(n)} \lda_\infty^n\,.
\end{equation}
As always, in practice we use the identification \eqref{identification2} (identifying the action and coadjoint actions accordingly) and the (co)adjoint orbit of an element $f\in \iota_\lambda \mathcal{F}_Q(\g)$ can be seen to be given by the elements
\begin{eqnarray}
\label{formula_orbit_elem}
    F=\Pi_-(\text{Ad}_{\varphi_+}\cdot f)\equiv \iota_\lda L\,.
\end{eqnarray}
In \eqref{formula_orbit_elem}, the adjoint action of $\varphi_+$ on $f$ is defined component-wise
\begin{equation}
    (\text{Ad}_{\varphi_+}\cdot f)_r(\lda_r)=\varphi_{r+}(\lda_r)\, f_r(\lda_r)\,\varphi_{r+}(\lda_r)^{-1}\,,~~r\in S\,.
\end{equation}
Thus, we have a construction that allows us to interpret a rational Lax matrix $L(\lda)$ as an element of a (co)adjoint orbit and recast \eqref{generating_Gaudin_eqs} as the following Lax equation in $\iota_\lambda \mathcal{F}_Q(\g)$
\begin{eqnarray}
\label{embedded_Lax}
    \partial_{t_k^r}\iota_\lda L=[R_\pm\nabla H_{k, r}(\iota_\lda L),\iota_\lda L]\,,
\end{eqnarray}
where $H_{k, r}$ are the following invariant functions on $\Lg_Q$
\begin{equation}
\label{Hamiltonians}
    H_{k, r}:\, X\in\Lg_Q\mapsto \res_{\lda_r=0}\frac{\Tr(X_r(\lda_r)^{k+1})}{k+1}\,,~~k\ge 1\,.
\end{equation}
We now apply the described framework to show how \eqref{generating_Gaudin_eqs} is derived in this context for $k=1,2$. Then we construct explicitly the corresponding Lagrangian coefficients of our multiform and check that their Euler-Lagrange equations produce the correct equations of motion.

\paragraph{Lax matrix and Lax equations for the first two flows:}   
Let us choose 
\begin{equation}
\Lambda(\lda) =\sum_{r=1}^N\frac{\Lambda_r}{\lda-\zeta_r}+\Omega \,, 
\end{equation}
and apply \eqref{formula_orbit_elem} to $f=\iota_\lda \Lambda$
to get 
\begin{equation}\label{orbit_L}
\begin{split}
\iota_\lda L=\Pi_-\left(\text{Ad}_{\varphi_+}\cdot \iota_\lda \Lambda\right)&= \iota_\lda\circ \pi_\lda\circ P_-\left(\text{Ad}_{\varphi_+}\cdot \iota_\lda \Lambda\right)\\[1ex] 
&=   \iota_\lda\circ \pi_\lda\left(\frac{\phi_1^{(0)}\,\Lambda_1\,(\phi_1^{(0)})^{-1}}{\lda-\zeta_1},\dots, \frac{\phi_N^{(0)}\,\Lambda_N\,(\phi_N^{(0)})^{-1}}{\lda-\zeta_N},\Omega \right)  \\[.5ex] 
 &\equiv\iota_\lda\circ \pi_\lda \left(\frac{A_1}{\lda-\zeta_1},\dots, \frac{A_N}{\lda-\zeta_N},\Omega \right) \\
&=\iota_\lda \left(\sum_{r=1}^N\frac{A_r}{\lda-\zeta_r}+\Omega\right)\,.
\end{split} 
\end{equation}
This is the desired form of \eqref{Lax_matrix_Gaudin} where now each $X_r$ is of the form $A_r=\phi_r^{(0)}\,\Lambda_r\,(\phi_r^{(0)})^{-1}$ with $\Lambda_r\in\g$ fixed and $\phi_r^{(0)}$ containing the dynamical degrees of freedom. This is the (co)adjoint description required to compute our Lagrangian coefficients, see below. 

Next, we derive the Lax equations in $\iota_\lambda \mathcal{F}_Q(\g)$ associated with the functions $H_{k, r}(\iota_\lda L)$ for $k=1, 2$. The gradient of $H_{k, r}$ at the point $\iota_\lda L$ is defined as the element of $\Lg_Q$ satisfying
\begin{equation}
\label{def_gradient}
    \lim_{\epsilon\to 0}\frac{H_{k, r}(\iota_\lda L+\epsilon \eta)-H_{k, r}(\iota_\lda L)}{\epsilon}=\langle \eta\,,\, \nabla H_{k, r}(\iota_\lda L)\rangle\,,
\end{equation}
for all $\eta \in \Lg_Q$. It is enough for our purposes to calculate $R_-(\nabla H_{k, r}(\iota_\lda L))$, therefore, we can restrict $\eta$ to $\Lg_{Q+}$. Thus, writing
\begin{equation}
    \nabla H_{k, r}(\iota_\lda L)=N^{(k)}+\iota_\lda h^{(k)}\,,\qquad N^{(k)}\in\Lg_{Q+}\,,~~h^{(k)}(\lda)\in \mathcal{F}_Q(\g)\,,
\end{equation}
recalling that $\Lg_{Q+}$ and $\iota_\lda \mathcal{F}_Q(\g)$ are isotropic with respect to the bilinear form in \eqref{eq:bf-lg-q}, \eqref{def_gradient} becomes
\begin{equation}
    \res_{\lda_r=0}\Tr\left( \eta_r \,\iota_{\lda_r} L^k \right)=\sum_{s \in S}\res_{\lda_s=0}
    \Tr \left( \eta_s \,\iota_{\lda_s} h^{(k)} \right),
\end{equation}
for any $\eta_s \in \Lg_{s+}$, $s\in S$, implying 
\begin{align}
    ( \iota_{\lda_s}h^{(k)})_- &= 0\,,~~\forall\,s\neq r\,,\\[1ex] 
    (\iota_{\lda_r}h^{(k)})_- &= (\iota_{\lda_r} L^k)_-\,.
\end{align}
This means that the rational function $h^{(k)}(\lda)$ has a nonzero principal only at $\zeta_r$ which equals $(\iota_{\lda_r} L^k)_-$, so
\begin{equation}
    h^{(k)}(\lda) = (\iota_{\lda_r} L^k)_-(\lda)\,,
\end{equation}
and we find
\begin{equation}
    R_-(\nabla H_{k, r}(\iota_\lda L)) = -\Pi_-(\nabla H_{k, r}(\iota_\lda L))= -\iota_\lda h^{(k)} = -\iota_\lda \left( (\iota_{\lda_r} L^k)_- \right)\,.
\end{equation}
    For $k=1, 2$, this gives us
\begin{eqnarray}
    R_-(\nabla H_{1, r}(\iota_\lda L)) = -\iota_\lda \frac{A_r}{\lda-\zeta_r}\,,
\end{eqnarray}
and
\begin{eqnarray}
    R_-(\nabla H_{2, r}(\iota_\lda L)) = -\iota_\lda \left( \frac{A_r^2}{(\lda-\zeta_r)^2} +\sum_{s \neq r} \frac{A_rA_s+A_sA_r}{(\lda - \zeta_r)(\zeta_r - \zeta_s)} + \frac{A_r\Omega +\Omega A_r}{\lda - \zeta_r} \right)\,, 
\end{eqnarray}
respectively.
As a consequence, we find the Lax equations for the two levels of flows as
\begin{eqnarray}
    \partial_{t_1^r}\iota_\lda L=\left[-\iota_\lda \frac{A_r}{\lda-\zeta_r}\,,\,\iota_\lda L\right]\,,
\end{eqnarray}
\begin{eqnarray}\label{eq:gaudin-Laxeqtwo}
    \partial_{t_2^r}\iota_\lda L=\left[-\iota_\lda \left(\frac{A_r^2}{(\lda-\zeta_r)^2} +\sum_{s \neq r} \frac{A_rA_s+A_sA_r}{(\lda - \zeta_r)(\zeta_r - \zeta_s)} + \frac{A_r\Omega +\Omega A_r}{\lda - \zeta_r} \right)\,,\,\iota_\lda L\right]\,.
\end{eqnarray}
Explicitly, they yield the following equations on the $A_s$,
\begin{equation}
\label{flow1}
    \begin{split}
  &   \partial_{t_1^r}A_s=  \frac{[A_r\,,\,A_s]}{\zeta_r-\zeta_s}\,,~~s\neq r\,,\\
   & \partial_{t_1^r}A_r=  -\sum_{s\neq r}\frac{[A_r\,,\,A_s]}{\zeta_r-\zeta_s}- [A_r\,,\,\Omega]\,,
        \end{split}
\end{equation}
thus reproducing \eqref{Gaudin_1}-\eqref{Gaudin_2} (\eqref{Gaudin_3} is automatic here since $\Omega$ is a constant element of $\g$), and 
\begin{equation}
\label{flow2}
    \begin{split}
        \partial_{t_2^r}A_s &= - \frac{[A_r^2\,,\,A_s]}{(\zeta_r-\zeta_s)^2} + \sum_{{s^\prime} \neq r} \frac{[A_r A_{s^\prime} + A_{s^\prime}A_r, A_s]}{(\zeta_r-\zeta_s)(\zeta_r-\zeta_{s^\prime})} + \frac{[A_r\Omega + \Omega A_r\,,\,A_s]}{\zeta_r-\zeta_s} \,,~~ s \neq r\,,\\
        \partial_{t_2^r}A_r &= \sum_{s \neq r} \frac{[A_r^2\,,\,A_s]}{(\zeta_r-\zeta_s)^2} -\sum_{s\neq r}\sum_{{s^\prime}\neq r}\frac{[A_r\,,\,A_s A_{s^\prime}]}{(\zeta_r-\zeta_s)(\zeta_r-\zeta_{s^\prime})} - \sum_{s \neq r} \frac{[A_r\,,\, A_s \Omega + \Omega A_s]}{\zeta_r-\zeta_s} - [A_r\,,\, \Omega^2]\,.
        \end{split}
\end{equation}

\paragraph{Lagrangian description:} Applying our formula for the Lagrangian coefficients, we obtain the following multiform
on the orbit of $\Lambda(\lda)$, with elements $\iota_\lda L$ given in \eqref{orbit_L},
\begin{equation}
    \Lag = \sum_{k=1}^{N} \sum_{r \in S} \Lag_{k, r}\, dt_k^r \,, 
\end{equation}
with
\begin{equation}
    \Lag_{k, r} = \sum_{s\in S} \res_{\lambda_s = 0} \Tr\left(\iota_{\lda_s}L\, \partial_{t_k^r}\varphi_{s+}(\lda_s)\,\varphi_{s+}(\lda_s)^{-1}\right)-H_{k, r}(\iota_\lda L)\,,
\end{equation}
where $H_{k, r}(\iota_\lda L)$ is the restriction of $H_{k, r}$ to $\iota_{\lda}L$. For the kinetic part, we have
\begin{equation}
    \res_{\lambda_s = 0} \Tr(\iota_{\lda_s}L\, \partial_{t_k^r}\varphi_{s+}(\lda_s)\,\varphi_{s+}(\lda_s)^{-1})=\Tr\left(\Lambda_s (\phi_s^{(0)})^{-1}\partial_{t_k^r}\phi_s^{(0)}\right)\,,~~s=1,\dots,N,
\end{equation}
and
\begin{equation}
    \res_{\lambda_\infty = 0} \Tr(\iota_{\lda_\infty}L\, \partial_{t_k^r}\varphi_{\infty+}(\lda_\infty)\,\varphi_{\infty+}(\lda_\infty)^{-1})=\Tr\left(\Omega \partial_{t_k^r}\phi_\infty^{(1)}\,\phi_\infty^{(1)}\right)=\frac{1}{2}\partial_{t_k^r}\Tr\left(\Omega (\phi_\infty^{(1)})^2\right).
\end{equation}
The contribution at $\infty$ is a total derivative, so it will not enter the Euler-Lagrange equations and hence we discard it. Thus, only the term $\phi_s^{(0)}$ in the Taylor series of $\varphi_{s+}(\lda_s)$ appears in the kinetic term. We will simply denote it by $\phi_s$ to lighten notations. 
The Lagrangian coefficients of the Gaudin multiform take the form
\begin{equation}
    \Lag_{k, r} = \sum_{s=1}^N\Tr\left(\Lambda_s \phi_s^{-1}\partial_{t_k^r}\phi_s\right) - H_{k, r}(\iota_\lda L).
\end{equation}
More explicitly, for $k=1,2$, we have
\begin{equation}
    H_{1, r}(\iota_{\lda}L) = \sum_{s\neq r}\frac{\Tr (A_rA_s)}{\zeta_r-\zeta_s}+\Tr(A_r\Omega)\,,
\end{equation}
and 
\begin{equation}
      H_{2, r}(\iota_{\lda}L) = \Tr\left(A_r \left(\sum_{s\neq r}\frac{A_s}{\zeta_r-\zeta_s}  +\Omega\right)^2\right) - \Tr\left(A_r^2 \left(\sum_{s\neq r}\frac{A_s}{(\zeta_r-\zeta_s)^2}\right) \right)\,.
\end{equation}
Varying $\Lag_{1, r}$ and $\Lag_{2, r}$ with respect to $\phi_s$, $s=1,\dots,N$ (recalling that $A_s=\phi_s\,\Lambda_s\,\phi_s^{-1}$), one can check by direct calculations that the Euler-Lagrange equations give exactly \eqref{flow1}-\eqref{flow2}.

\begin{remark}
    The algebraic framework we have used to describe the Lagrangian multiform for the Gaudin is to a very large extent similar to that used in \cite{CSV} to construct Lagrangian multiforms of Zakharov-Mikhailov type. Therefore, in hindsight, it is perhaps not so surprising that the Lagrangian 
    \begin{equation}
\label{Lag_1k}
    \Lag_{1, r} = \sum_{s=1}^N\Tr\left(\Lambda_s  \phi_s^{-1}\partial_{t_1^r}\phi_s\right) - \sum_{s\neq r}\frac{\Tr (A_rA_s)}{\zeta_r-\zeta_s} - \Tr(A_r\Omega)\,,
\end{equation}
appears to be the direct analogue in the finite-dimensional case of the Zakharov-Mikhailov Lagrangians which describe integrable field theories with rational Lax matrices \cite{ZM}. It is a rather satisfying outcome that we have unravelled the unifying structure underlying such Lagrangians, whether in finite or infinite dimensions. They are all connected to Lie dialgebras which control the structure of their kinetic part and tell us which potentials to include (invariant functions on $\g^*$). Note that in \cite{VW}, a very similar Lagrangian, their Equation (24), was constructed by a completely different method: an adaptation of the idea of 4d Chern-Simons theory, see \cite{CY} and references therein, and of the construction in \cite{CSV2} to the case of a BF theory in 3d. This suggests the tantalising direction of deriving our Lagrangian multiforms from an appropriately adapted BF theory. This could perhaps offer an interpretation for the appearance of Lie dialgebras from this point of view, instead of introducing them ad hoc as we do in the present paper. 
\end{remark}

We know from the general theory that the closure relation $d\Lag = 0$ holds on shell. This implies
\begin{equation}
\label{eq:gaudin-closure}
    \partial_{t_j^s} \Lag_{k, r} - \partial_{t_k^r} \Lag_{j, s} = 0\,,
\end{equation}
for all possible combinations of $j,k$ and $r,s$. As we know, the kinetic and potential contributions give zero separately in each case. Let us illustrate the main steps here for $k=1$, $j=2$ and $r\neq s$ in \eqref{eq:gaudin-closure}, the left-hand side of which will then read
\begin{align*}
   \sum_{s^\prime=1}^N \partial_{t_2^s} \Tr\left(\Lambda_{s^\prime} \phi_{s^\prime}^{-1}\partial_{t_1^r}\phi_{s^\prime}\right) - \sum_{s^\prime=1}^N \partial_{t_1^r} \Tr\left(\Lambda_{s^\prime} \phi_{s^\prime}^{-1}\partial_{t_2^s}\phi_{s^\prime}\right)
   - \partial_{t_2^s} H_{1, r}(\iota_\lda L) + \partial_{t_1^r} H_{2, s}(\iota_\lda L)\,.
\end{align*}
Using the equations of motion, we have
\begin{align*}
    &\partial_{t_2^s} H_{1, r}(\iota_{\lda}L)\\
    &= \sum_{s^\prime\neq r} \frac{1}{\zeta_r-\zeta_{s^\prime}} \Tr \left(\left(  - \frac{[A_s^2\,,\,A_r]}{(\zeta_s-\zeta_r)^2} + \sum_{{s^{\prime\prime}} \neq s} \frac{[A_s A_{s^{\prime\prime}} + A_{s^{\prime\prime}}A_s, A_r]}{(\zeta_s-\zeta_r)(\zeta_s-\zeta_{s^{\prime\prime}})} +  \frac{[A_s \Omega + \Omega A_s\,,\,A_r]}{\zeta_s-\zeta_r} \right) A_{s^\prime} \right)\\
    &~~+ \sum_{\substack{s^\prime \neq r\\s^\prime \neq s}} \frac{1}{\zeta_r-\zeta_{s^\prime}} \Tr \left(A_r \left(  - \frac{[A_s^2\,,\,A_{s^\prime}]}{(\zeta_s-\zeta_{s^\prime})^2} + \sum_{{s^{\prime\prime}} \neq s} \frac{[A_s A_{s^{\prime\prime}} + A_{s^{\prime\prime}}A_s, A_{s^\prime}]}{(\zeta_s-\zeta_{s^\prime})(\zeta_s-\zeta_{s^{\prime\prime}})} +  \frac{[A_s \Omega + \Omega A_s\,,\,A_{s^\prime}]}{\zeta_s-\zeta_{s^\prime}}  \right) \right)\\
    &~~+ \frac{1}{\zeta_r-\zeta_s} \Tr \Bigg(A_r \Bigg( \sum_{s^\prime \neq s} \frac{[A_s^2\,,\,A_{s^\prime}]}{(\zeta_s-\zeta_{s^\prime})^2} -\sum_{s^\prime \neq s}\sum_{{s^{\prime\prime}}\neq s}\frac{[A_s\,,\,A_{s^\prime} A_{s^{\prime\prime}}]}{(\zeta_s-\zeta_{s^\prime})(\zeta_s-\zeta_{s^{\prime\prime}})}\\
    &\hspace{50ex} -   \sum_{s^\prime \neq s} \frac{[A_s\,,\, A_{s^\prime} \Omega + \Omega A_{s^\prime}]}{\zeta_s-\zeta_{s^\prime}}  - [A_s\,,\, \Omega^2] \Bigg) \Bigg)\\
    &~~+\Tr \left(\left( - \frac{[A_s^2\,,\,A_r]}{(\zeta_s-\zeta_r)^2} +  \sum_{{s^{\prime}} \neq s} \frac{[A_s A_{s^{\prime}} +  A_{s^{\prime}}A_s, A_r]}{(\zeta_s-\zeta_r)(\zeta_s-\zeta_{s^{\prime}})}  + \frac{[A_s \Omega + \Omega A_s\,,\,A_r]}{\zeta_s-\zeta_r} \right) \Omega \right)\,. 
\end{align*}
This is seen to add up to zero by assembling the terms of the same nature (quartic, cubic or quadratic in $A$), manipulating the sums, using the ${\rm ad}$-invariance property $\Tr([A,B]C)=\Tr(A[B,C])$ and the identity
$$\frac{1}{(\zeta_r-\zeta_s)(\zeta_r-\zeta_{s^\prime})}+\frac{1}{(\zeta_s-\zeta_{s^\prime})(\zeta_r-\zeta_{s^\prime})} +\frac{1}{(\zeta_s-\zeta_r)(\zeta_s-\zeta_{s^\prime})}=0 \,.$$
Similar calculations give $\partial_{t_1^r} H_{2, s}(\iota_\lda L)=0$.
For the kinetic terms we have
\begin{align*}
    &\partial_{t_2^s} \sum_{s^\prime=1}^N \Tr\left(\Lambda_{s^\prime} \phi_{s^\prime}^{-1}\partial_{t_1^r}\phi_{s^\prime}\right) - \partial_{t_1^r} \sum_{s^\prime=1}^N \Tr\left(\Lambda_{s^\prime} \phi_{s^\prime}^{-1}\partial_{t_2^s}\phi_{s^\prime}\right)\\ 
    &= \sum_{s^\prime=1}^N \Tr \left((\partial_{t_2^s}A_{s^\prime}) (\partial_{t_1^r}\phi_{s^\prime}) \phi_{s^\prime}^{-1} \right) - \sum_{s^\prime=1}^N \Tr \left((\partial_{t_1^r}A_{s^\prime}) (\partial_{t_2^s}\phi_{s^\prime}) \phi_{s^\prime}^{-1} \right)\\
    &~~+ \sum_{s^\prime=1}^N \Tr \left( A_{s^\prime} [(\partial_{t_2^s}\phi_{s^\prime})\phi_{s^\prime}^{-1}\,, (\partial_{t_1^r}\phi_{s^\prime})\phi_{s^\prime}^{-1}] \right)\\
    &~~+ \sum_{s^\prime=1}^N \Tr \left( A_{s^\prime} \left((\partial_{t_2^s}\partial_{t_1^r}\phi_{s^\prime})\phi_{s^\prime}^{-1} - (\partial_{t_1^r}\partial_{t_2^s}\phi_{s^\prime})\phi_{s^\prime}^{-1}\right) \right).
\end{align*}
The commutativity of flows ensures that the last term equals zero. Further, using the relation
\begin{equation}\label{eq:gaudin-flowrearrange}
    \partial_{t_2^s}A_{s^\prime} = [(\partial_{t_2^s} \phi_{s^\prime})\phi_{s^\prime}^{-1}\,, A_{s^\prime}],~~s^\prime = 1,\dots,N\,,
\end{equation}
it is easy to see that the first and the third terms cancel each other. Finally, for the second term, using ${\rm ad}$-invariance, \eqref{eq:gaudin-flowrearrange} and the on-shell relations in \eqref{flow1} and \eqref{flow2}, we have
\begin{align*}
    &\sum_{s^\prime=1}^N \Tr \left((\partial_{t_1^r}A_{s^\prime}) (\partial_{t_2^s}\phi_{s^\prime}) \phi_{s^\prime}^{-1} \right)\\
    &= \Tr \left((\partial_{t_1^r}A_r) (\partial_{t_2^s}\phi_r) \phi_r^{-1} \right) + \sum_{s^\prime \neq r} \Tr \left((\partial_{t_1^r}A_{s^\prime}) (\partial_{t_2^s}\phi_{s^\prime}) \phi_{s^\prime}^{-1} \right)\\
    &= -\sum_{s^\prime \neq r} 
    \Tr \left( \frac{[A_r\,,\,A_{s^\prime}]}{\zeta_r - \zeta_{s^\prime}} (\partial_{t_2^s}\phi_r) \phi_r^{-1} \right) - \Tr \left([A_r\,,\,\Omega] (\partial_{t_2^s}\phi_r) \phi_r^{-1} \right)
    \\    &~~
    + \sum_{s^\prime \neq r} \Tr \left(\frac{[A_r\,,\,A_{s^\prime}]}{\zeta_r - \zeta_{s^\prime}} (\partial_{t_2^s}\phi_{s^\prime}) \phi_{s^\prime}^{-1} \right)\\
    &= -\sum_{s^\prime \neq r} \frac{\Tr (A_{s^\prime} \partial_{t_2^s} A_r)}{\zeta_r - \zeta_{s^\prime}} - \Tr(\Omega \partial_{t_2^s}A_r) -\sum_{s^\prime \neq r} \frac{\Tr (A_r \partial_{t_2^s} A_{s^\prime})}{\zeta_r - \zeta_{s^\prime}}\\
    &=-\partial_{t_2^s} H_{1, r}(\iota_{\lda}L)
\end{align*}
which we previously showed to be zero.

\section{Conclusion}\label{ccl}

In this work, we provided an answer to the problem of constructing all the coefficients in Lagrangian $1$-form for a large class of finite-dimensional integrable systems (any model fitting the Lie dialgebra framework). A reinterpretation of our construction is that we proved that any collection of compatible equations in the Lax form 
\begin{equation}
    \partial_{t_k}L=[R_\pm \nabla H_k(L),L]\,,~~ k=1,\dots,N\,,
\end{equation}
is {\it variational}, by explicitly providing a collection of Lagrangians assembled in a multiform. The closure relation is equivalent to the involutivity of the Hamiltonians $H_k$. This is a corollary of our stronger result, Theorem \ref{prop_double_zero}. 

We recast our construction in a more general context which makes it clear how it descends from a ``free'' (phase space) Lagrangian on the cotangent bundle of a Lie group by reduction. This procedure is well-known in the Hamiltonian framework and we have explained how it translates into our framework, by exploiting the correspondence between moment maps and Noether currents. This offers a larger perspective on our results. On the one hand, it may lead to the possibility of constructing Lagrangian multiforms for systems of Calogero-Moser type by using reduction ideas appropriately. From the point of view of $r$-matrices, a strong motivation, and at the same time an anticipated difficulty, is the appearance of dynamical $r$-matrices in such systems. Comparison with the early work on Calogero-Moser multiforms \cite{YKLN} would be beneficial. On the other hand, it shows that our Lagrangian coefficients turn out to have a structure similar to those appearing in so-called geometric actions. The latter can be traced back (at least) to~\cite{W,AFS, AS,ANPZ} and are concerned with quantisation using Feynman's path integral in conjunction with coadjoint orbit methods. This interesting connection deserves further investigation.

Of the many models we could have used to illustrate the results in the present work, we chose the open Toda chain and the Gaudin model, two emblematic finite-dimensional integrable systems. The motivation for studying the finite Gaudin model is Vicedo's construction of a class of non-ultralocal field theories as affine Gaudin models \cite{Vic}. We very much hope that the present results combined with the approach of \cite{CSV} and Vicedo's construction will allow us to overcome the current limitation of Lagrangian multiforms to only ultralocal field theories.

\section*{Acknowledgements} V.C. wishes to thank T. Skrypnyk for helpful discussions on higher Gaudin Hamiltonians and B. Vicedo and M. Vermeeren for helpful feedback and comments on our draft. V.C. is grateful to M. Semenov-Tian-Shansky for pointing out that the associative YBE in Remark \ref{rem_assoc} was already introduced in \cite{STS2} and is related to the older notion of Rota-Baxter algebras. We want to express our special gratitude to one referee whose insightful comments led to the content of Section \ref{reduction}. V.C. and M.D. would like to thank the Isaac Newton Institute for Mathematical Sciences for support and hospitality during the programme {\it Dispersive Hydrodynamics} when initial work on this paper was undertaken (EPSRC Grant Number EP/R014604/1). A.A.S. is funded by the School of Mathematics EPSRC Doctoral Training Partnership Studentship (Project Reference Number 2704447).

\section*{Conflict of interest statement}
On behalf of all authors, the corresponding author states that there is no conflict of interest.

\section*{Data availability statement}
This manuscript has no associated data.

\bibliographystyle{sn-aps-M}
\bibliography{biblio}

\begin{comment} 

\end{comment}

\end{document}